\documentclass[preprint,3p]{elsarticle}
\usepackage{amssymb,amsmath,enumerate,graphicx,hyperref,amsfonts,latexsym,url,wasysym,mathtools,amsthm}
\usepackage[utf8]{inputenc}
\usepackage{natbib}
\theoremstyle{plain}
\newtheorem{Thm}{Theorem}
\newtheorem{lemma}[Thm]{Lemma}
\newtheorem{cor}[Thm]{Corollary}
\newtheorem{prop}[Thm]{Proposition}
\theoremstyle{definition}
\newtheorem{rem}[Thm]{Remark}
\newtheorem{exm}[Thm]{Example}

\usepackage{tikz}
\usetikzlibrary{snakes,shapes}
\tikzstyle{tre}=[circle,draw,minimum size=3mm,inner sep=1pt]
\tikzstyle{trep}=[circle,draw,minimum size=2mm,inner sep=0.1pt]
\tikzstyle{trepp}=[circle,draw,minimum size=1.5mm,inner sep=0.1pt]
\newcommand{\etq}[1]{%
\draw (#1) node {\scriptsize $#1$};
}

\renewcommand{\leq}{\leqslant}
\renewcommand{\geq}{\geqslant}
\newcommand{\NN}{\mathbb{N}}

\newcommand{\red}[1]{\textcolor{red}{#1}}

\newcommand{\oC}{\overline{C}}

\newcommand{\TT}{\mathcal{T}}

\DeclarePairedDelimiter\ceil{\lceil}{\rceil}
\DeclarePairedDelimiter\floor{\lfloor}{\rfloor}

\begin{document}
\begin{frontmatter}

\title{The minimum value of the Colless index}

\author{Tom\'as M. Coronado}
\ead{t.martinez@uib.es}
\author{Francesc Rossell\'o}
\ead{cesc.rossello@uib.es}
\address{Dept. of Mathematics and Computer Science, University of the Balearic Islands, E-07122 Palma, Spain, and Balearic Islands Health Research Institute (IdISBa), E-07010 Palma, Spain}

\begin{abstract}
\red{ATTENTION: This manuscript has been subsumed by another manuscript, which can be found on arXiv: arXiv:1907.05064}\\
The Colless index is one of the oldest and most widely used balance indices for rooted bifurcating  trees.  Despite its popularity, its minimum value on the space $\mathcal{T}_n$ of rooted bifurcating  trees with $n$ leaves is only known when $n$ is a power of 2. In this paper we fill this gap in the literature, by providing a formula that computes, for each $n$, the minimum Colless index on $\mathcal{T}_n$, and characterizing those trees where this minimum value is reached.
\end{abstract}

\end{frontmatter}

\section{Introduction}

\noindent One of the main goals of evolutionary biology is to understand what factors influence evolutionary processes and their effect. Since phylogenetic trees are the standard representation of joint evolutive histories of groups of species, it is natural to look for the imprint of these factors in the shapes of phylogenetic trees \citep{Mooers97,Stich09}. This has motivated the introduction of various indices that quantify topological features of tree shapes supposedly related to properties of the evolutionary processes represented by the trees. These indices have been   used then to test evolutionary hypothesis  \citep{Blum05,KiSl:93,Mooers97,Stam02}
and to compare tree shapes \citep{Avino18,Goloboff17}, among other applications. Since the early observation by \citet{Yule} that taxonomic trees tend to be asymmetric, with many small clades and a few large ones at every taxonomic level, the most popular topological feature used to describe the shape of a phylogenetic tree  has been its \emph{balance}, the tendency
of the children of any given node to have the same number of descendant leaves. The \emph{imbalance} of a phylogenetic tree reflects   the propensity of evolutive events to occur along specific lineages  \citep{Nelson}.

Several \emph{balance indices} have been proposed so far to quantify the balance (or actually, in most cases, the imbalance) of a phylogenetic tree; see, for instance, \citep{Colless:82,CMR,KiSl:93,MRR1,Sackin:72,Shao:90,cherries} and the section ``Measures of overall asymmetry'' in Felsenstein's book \citep{fel:04} (pp.  562--563). Among them, the Colless index, introduced by D. H. Colless \citep{Colless:82}, is one of the oldest and most popular. It is defined, on a bifurcating tree $T$, as  the sum, over all the internal nodes $v$ of $T$, of the absolute value of the difference of the numbers of descendant leaves of the pair of children of $v$. Although it was defined primarily for bifurcating phylogenetic trees,  \citet{Shao:90} proposed to extend it to multifurcating trees by  taking into account in this sum only the bifurcating  nodes; a more meaningful extension to arbitrary trees has been proposed recently by \citet{MRR2}. Anyway, in this paper we deal with the classical Colless index for bifurcating trees, and since it is a \emph{shape index}, in the sense that it does not depend on the actual labels of the tree's leaves, we shall not take into account these labels and we shall consider this index defined on unlabeled rooted bifurcating trees, or simply  \emph{bifurcating trees} for short.

By its very definition, the Colless index measures directly the total amount of imbalance of a bifurcating tree, with the notion of balance recalled above. And, indeed, for every number $n$ of leaves, it is folklore knowledge that the bifurcating tree with $n$ leaves and largest Colless index is the \emph{comb}: the bifurcating tree such that each internal node has a leaf child, considered since the early paper by \citet{Sackin:72} to be the most imbalanced type of phylogenetic tree. This fact was already hinted by \citet{Colless:82}, but he gave a wrong value for the Colless index of a comb which was later corrected by \citet{Heard92}, giving the correct maximum value of $(n-1)(n-2)/2$. Actually, to our knowledge, no explicit direct proof of the maximality of this Colless value has been provided in the literature, but it can be easily deduced as a particular case of Thm. 18 in \citep{MRR2}.

But, in spite of its popularity and wide use,  the minimum Colless index of a bifurcating tree with $n$ leaves remains unknown beyond the 
 often stated straighforward result that for numbers of leaves that are powers of 2 it is reached  at the fully symmetric bifurcating trees, which clearly have Colless index 0; see for example \citep{Heard92,KiSl:93,Mooers97}. To have a closed formula for this minimum value is necessary, for instance, in order to normalize the Colless index to the range $[0,1]$  for every number of leaves, making its value independent of its size as it is recommended by \citet{Shao:90} or \citet{Stam02}. 
 
 The goal of this paper is then to fill this gap  in the literature.
In  Sections \ref{sec:CB} and \ref{sec:min}  we prove that, for every $n$, the minimum Colles index on the space $\TT_n$ of bifurcating trees with $n$ leaves is reached at the \emph{maximally balanced  trees}, those trees such that, for every internal node $v$, the numbers of descendant leaves of its two children differ at most of 1, and we provide a formula that computes  this minimum value for each $n$ (see Proposition \ref{prop:valueC}).  This  result is not surprising, because these maximally balanced  trees are considered, as their name hints, the ``most balanced binary trees'' \citep{Shao:90}, and many other balance indices reach their smallest values at them: for instance,  \citet{Fischer19} has recently proved it for the Sackin index introduced in \citep{Sackin:72,Shao:90}, and it is also the case for the total cophenetic index \citep{MRR1} or the rooted quartet index \citep{CMR} (in fact, the   maximally balanced  trees have the largest rooted quartet index, because this index, unlike the others, grows from most imbalanced to most balanced).  Unfortunately, there are trees that are not maximally balanced but that  also yield the minimum Colless index for their number of leaves (see, for instance, Fig. \ref{fig:min6} below). Then, in Section \ref{sec:QB} we characterize the trees in $\TT_n$ having minimum Colless index and we provide an efficient algorithm that  produces them for any $n$.\smallskip

\noindent\textbf{Notations.} 
 In this paper, by a  \emph{tree} we mean a  rooted tree, understood as a directed graph with its arcs pointing away from the root. Such a tree is \emph{bifurcating}  when all its nodes have  out-degree either 0 (its \emph{leaves}) or 2 (its \emph{internal nodes}). We shall denote by $\TT_n$ the set  of (isomorphism classes of) bifurcating trees with $n$ leaves.  A \emph{cherry} is a tree consisting only of a root and two leaves.

Given a bifurcating tree $T\in \TT_n$, we shall denote by $V(T)$ its set of nodes and by $V_{int}(T)$ its set of internal nodes.  For every $u,v\in V(T)$,  $v$ is a \emph{child} of $u$ when $T$ contains the arc $(u,v)$, and $v$ is a  \emph{descendant} of $u$ when $T$ contains  a (directed) path from $u$ to $v$. 
 For every $v\in V(T)$, the \emph{subtree $T_v$ of $T$ rooted at $v$} is the subgraph of $T$ induced by the  descendants of $v$ and $\kappa_T(v)$ is the number of leaves of $T_v$.   Given two bifurcating trees $T\in \TT_{n_1}$ and $T'\in \TT_{n_2}$, their \emph{root join} is   the tree $T\star T'\in \TT_{n_1+n_2}$ whose subtrees rooted at the children of the root are $T$ and $T'$.

For every $v\in V_{int}(T)$, say with children $v_1,v_2$,  the \emph{balance value} of $v$ is $bal_T(v)=|\kappa_T(v_1)-\kappa_T(v_2)|$. We shall say that an internal node of $T$ is \emph{balanced} when its balance value is 0, and \emph{imbalanced} otherwise.  The \emph{Colless index} \citep{Colless:82}  of  $T\in \TT_n$ is the sum of the balance values of its internal nodes:
$$
C(T)=\sum_{v\in V_{int}(T)} bal_T(v).
$$
It is easy to see that the Colless index satisfies the following recurrence \citep{Rogers:93}: if $T=T_1\star T_2$, with $T_1\in \TT_{n_1}$ and $T_2\in \TT_{n_2}$, then
\begin{equation}
C(T)=C(T_1)+C(T_2)+|n_1-n_2|.
\label{eq:reccolless}
\end{equation}

\section{Colless indices of maximally balanced trees}\label{sec:CB}

\noindent An internal node $v$ in a bifurcating tree $T$ is \emph{balanced} when $bal_T(v)\leq 1$: i.~e., when  its two children have $\lceil \kappa_T(v)/2\rceil$  and $\lfloor \kappa_T(v)/2\rfloor$  descendant leaves, respectively. A tree $T\in \TT_n$ is \emph{maximally balanced} when all its internal nodes are balanced. Therefore, a tree is maximally balanced when its root is balanced and both subtrees rooted at the children of the root are maximally balanced. This easily implies, on the one hand, that, for any number $n$ of leaves, there is only one maximally balanced tree in $\TT_n$, which we shall denote by $B_n$, and, on the other hand, that  the maximally balanced trees satisfy the following recurrence:
\begin{equation}
B_n=B_{\ceil*{n/2}}\star B_{\floor*{n/2}}.
\label{eq:recB}
\end{equation}
When $n$ is a power of 2, $B_n$ is the \emph{fully symmetric bifurcating tree}, where, for each internal node, the pair of subtrees rooted at its children  are  isomorphic. Fig.~\ref{fig:bal} depicts the trees $B_6$, $B_7$ and $B_8$.   

\begin{figure}[htb]
\begin{center}
\begin{tikzpicture}[thick,>=stealth,scale=0.25]
\draw(0,0) node [tre] (1) {};  \etq 1
\draw(2,0) node [tre] (2) {};  \etq 2
\draw(4,0) node [tre] (3) {};  \etq 3
\draw(6,0) node [tre] (4) {};  \etq 4
\draw(8,0) node [tre] (5) {};  \etq 5
\draw(10,0) node [tre] (6) {};  \etq 6
\draw(1,2) node[tre] (a) {};
\draw(9,2) node[tre] (c) {};
\draw(2,4) node[tre] (b) {};
\draw(8,4) node[tre] (d) {};
\draw(5,6) node[tre] (r) {};
\draw  (a)--(1);
\draw  (a)--(2);
\draw  (b)--(a);
\draw  (b)--(3);
\draw  (c)--(5);
\draw  (c)--(6);
\draw  (d)--(4);
\draw  (d)--(c);
\draw  (r)--(b);
\draw  (r)--(d);
\end{tikzpicture}
\quad
\begin{tikzpicture}[thick,>=stealth,scale=0.25]
\draw(0,0) node [tre] (1) {};  \etq 1
\draw(2,0) node [tre] (2) {};  \etq 2
\draw(4,0) node [tre] (3) {};  \etq 3
\draw(6,0) node [tre] (4) {};  \etq 4
\draw(8,0) node [tre] (5) {};  \etq 5
\draw(10,0) node [tre] (6) {};  \etq 6
\draw(12,0) node [tre] (7) {};  \etq 7
\draw(1,2) node[tre] (a) {};
\draw(5,2) node[tre] (b) {};
\draw(3,4) node[tre] (c) {};
\draw(11,2) node[tre] (d) {};
\draw(9,4) node[tre] (e) {};
\draw(6,6) node[tre] (r) {};
\draw  (a)--(1);
\draw  (a)--(2);
\draw  (b)--(3);
\draw  (b)--(4);
\draw  (c)--(a);
\draw  (c)--(b);
\draw  (d)--(6);
\draw  (d)--(7);
\draw  (e)--(d);
\draw  (e)--(5);
\draw  (r)--(c);
\draw  (r)--(e);
\end{tikzpicture}
\quad
\begin{tikzpicture}[thick,>=stealth,scale=0.25]
\draw(0,0) node [tre] (1) {};  \etq 1
\draw(2,0) node [tre] (2) {};  \etq 2
\draw(4,0) node [tre] (3) {};  \etq 3
\draw(6,0) node [tre] (4) {};  \etq 4
\draw(8,0) node [tre] (5) {};  \etq 5
\draw(10,0) node [tre] (6) {};  \etq 6
\draw(12,0) node [tre] (7) {};  \etq 7
\draw(14,0) node [tre] (8) {};  \etq 8
\draw(1,2) node[tre] (a) {};
\draw(5,2) node[tre] (b) {};
\draw(3,4) node[tre] (c) {};
\draw(9,2) node[tre] (d) {};
\draw(13,2) node[tre] (e) {};
\draw(11,4) node[tre] (f) {};
\draw(7,6) node[tre] (r) {};
\draw  (a)--(1);
\draw  (a)--(2);
\draw  (b)--(3);
\draw  (b)--(4);
\draw  (c)--(a);
\draw  (c)--(b);
\draw  (d)--(5);
\draw  (d)--(6);
\draw  (e)--(7);
\draw  (e)--(8);
\draw  (f)--(d);
\draw  (f)--(e);
\draw  (r)--(c);
\draw  (r)--(f);
\end{tikzpicture}
\end{center}
\caption{\label{fig:bal} 
Three maximally balanced trees. The tree with 8 leaves is fully symmetric.}
\end{figure}
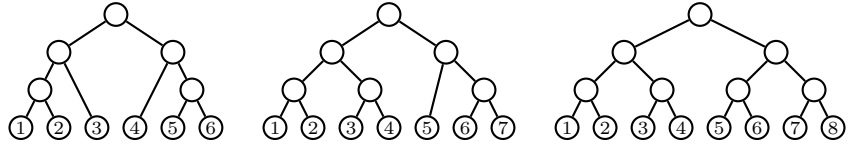

We shall denote by $C(n)$ the Colless index of $B_n$. The sequence $C(n)$ satisfies that $C(1)=0$ and, by (\ref{eq:reccolless}) and (\ref{eq:recB}), for every $n\geq 2$
\begin{equation}
C(n) =  C(\ceil*{n/2}) + C(\floor*{n/2}) + \ceil*{n/2} - \floor*{n/2}.
\label{eq:recC}
\end{equation}

The following well-known lemma is easily proved by induction on $n$, using recurrence (\ref{eq:reccolless}).

\begin{lemma}\label{lem:C=0}
For every $n\geq 1$ and $T\in \TT_n$, $C(T)=0$ if, and only if, $n$ is a power of 2 and $T=B_n$. \qed
\end{lemma}


Nex result gives a closed formula for $C(n)$ in terms of the binary expansion of $n$.

\begin{prop}\label{prop:valueC}
If $n=\displaystyle\sum_{j=1}^\ell 2^{m_j}$ with $m_1>\cdots>m_\ell$, then
$C(n)=\displaystyle\sum_{j=2}^\ell 2^{m_j}(m_1-m_j-2(j-2)).$
\end{prop}

\begin{proof}
For every  $n\geq 1$, let $\oC(n)=\sum_{j=2}^\ell 2^{m_j}(m_1-m_j-2(j-2))$, where $n=\sum_{j=1}^\ell 2^{m_j}$ with $m_1>\cdots>m_\ell$. Since $\oC(1)=\oC(2^0)=0=C(1)$, to prove that $\oC(n)=C(n)$ for every $n\geq 1$ it is enough to prove that the sequence $\oC(n)$ satisfies the recurrence (\ref{eq:recC}). Now, to prove the latter, we must distinguish two cases, depending on the parity of $n$.

Assume first that $n$ is even,  i.~e.,  $m_\ell>0$. In this case, $\floor*{n/2}=\ceil*{n/2}=\sum_{j=1}^\ell 2^{m_j-1}$ with $m_1-1>\cdots>m_\ell-1$  and then
$$
\begin{array}{l}
\displaystyle \oC(\ceil*{n/2}) + \oC(\floor*{n/2}) + \ceil*{n/2} - \floor*{n/2}= 2\oC\Big(\sum_{j=1}^\ell 2^{m_j-1}\Big)\\
\qquad  \displaystyle =2\Big(\sum_{j=2}^\ell 2^{m_j-1}\big(m_1-1-(m_j-1)-2(j-2)\big)\Big)=\sum_{j=2}^\ell 2^{m_j}(m_1-m_j-2(j-2))=\oC(n).
\end{array}
$$

Assume now that $n$ is odd, i.~e., $m_\ell=0$. Let $k=\min\{j\mid m_j=\ell-j\}$ (which exists because $m_\ell=\ell-\ell$). Then, $\floor*{n/2}=\sum_{j=1}^{\ell-1} 2^{m_j-1}$, with $m_1-1>\cdots>m_{\ell-1}-1$,
and
$$
\ceil*{n/2}=\sum_{j=1}^{\ell-1} 2^{m_j-1}+1=\sum_{j=1}^{k-1} 2^{m_j-1}+\sum_{j=k}^{\ell-1} 2^{\ell-j-1}+1=\sum_{j=1}^{k-1}2^{m_j-1}+2^{\ell-k}
$$
with $m_1-1>\cdots>m_{k-1}-1>\ell-k$. 
In this case, 
$$
\begin{array}{l}
\oC(\ceil*{n/2}) + \oC(\floor*{n/2}) + \ceil*{n/2} - \floor*{n/2}\\
\qquad \displaystyle= 
\sum_{j=2}^{k-1}2^{m_j-1}\big((m_1-1)-(m_j-1)-2(j-2))+2^{\ell-k}(m_1-1-(\ell-k)-2(k-2)\big)\\
\qquad\qquad \displaystyle+
\sum_{j=2}^{\ell-1}2^{m_j-1}\big((m_1-1)-(m_j-1)-2(j-2)\big)+1\\
\qquad \displaystyle= 
\sum_{j=2}^{k-1}2^{m_j-1}(m_1 -m_j-2(k-2))+2^{m_k}(m_1-m_k-2(k-2))-2^{\ell-k}\\
\qquad\qquad \displaystyle+
\sum_{j=2}^{k-1}2^{m_j-1}(m_1-m_j-2(j-2))+\sum_{j=k}^{\ell-1} 2^{\ell-j-1}(m_1-(\ell-j)-2(j-2))+1\\
\quad\mbox{(because $m_j=\ell-j$ for every $j\geq k$)}\\
\qquad \displaystyle= \sum_{j=2}^{k}2^{m_j}(m_1-m_j-2(j-2))+\sum_{j=k}^{\ell-1}2^{\ell-j-1}(m_1-(\ell-j)-2(j-2))+1-2^{\ell-k}\\
\qquad \displaystyle= \sum_{j=2}^{k}2^{m_j}(m_1-m_j-2(j-2))+
\sum_{i=k+1}^{\ell}2^{\ell-i}(m_1-(\ell-i)-2(i-2)+1)+1-2^{\ell-k}\\
\qquad \displaystyle= \sum_{j=2}^{k}2^{m_j}(m_1-m_j-2(j-2))+
\sum_{i=k+1}^{\ell}2^{m_i}(m_1-m_i-2(i-2))+\sum_{i=k+1}^{\ell}2^{\ell-i} +1-2^{\ell-k}\\
\qquad \displaystyle= \sum_{j=2}^{\ell}2^{m_j}(m_1-m_j-2(j-2))=\oC(n)
\end{array}
$$
This completes the proof that $\oC(n)$ satisfies (\ref{eq:recC}), and hence that $C(n)=\oC(n)$ for every $n\geq 1$.
\end{proof}

\begin{rem}
Since the balance of every internal node in a maximally balanced tree $B_n$ is at most 1, $C(n)$ is equal to the number of imbalanced nodes in $B_n$. Now, an internal node in $B_n$ is balanced if, and only if, the subtrees rooted at their children are isomorphic: that is,  with the notations in \citep{Ford}, if, and only if, it is a \emph{symmetric branch point}.
Therefore, the number of symmetric branch points in $B_n$ is $n-1-C(n)$, which implies, by Lemma 31 in \citep{Ford}, that the number of automorphisms of $B_n$ is $2^{n-1-C(n)}$. This says that the sequence $C(n)$ is equal to the sequence A296062 in Sloane's \textsl{On-Line Encyclopedia of Integer Sequences}  \citep{Sloane}.
So, Proposition \ref{prop:valueC} gives an explicit formula for that sequence.
\end{rem}

\section{The minimum Colless index}\label{sec:min}

\noindent In this section we prove that $C(n)$ is the minimum value of the Colless index on $\TT_n$.

\begin{lemma}\label{lem:key}
For every $(n, s)\in\NN^2$ with $n\geq 1$,
$C(n + s) + C(n) + s \geq C(2n + s)$.
\end{lemma}

\begin{proof}
We shall prove by induction on $n$ that, for every $n\geq 1$,  the inequality
\begin{equation}
C(n + s) + C(n) + s \geq C(2n + s)
\label{eq:goal}
\end{equation}
holds for every $s\geq 0$. Since $C(1)=0$, the base case $n=1$ says that, for every $s\geq 0$,
\begin{equation}
C(1+s) + s \geq C(2+s).
\label{eq:base1}
\end{equation}
We prove it by induction on $s$.
The cases  $s=0$ and $s=1$ are obviously true, because  $C(1) + 0 = 0 = C(2)$ and $C(2) + 1 = 1 = C(3)$. Let us consider now the case $s\geq 2$ and let us assume that $C(1+s') +  s'  \geq C(2+s')$ for every $s'< s$. 
To prove the induction step, we distinguish two cases.
\begin{itemize}
\item If $s\in2\NN$, say $s=2t$ with $t\geq 1$, then
$C(1+s) + s=C(2t+1)+2t=C(t+1)+C(t)+1+2t$ and
$C(2+s)=C(2t+2)=2C(t+1)$
and the desired inequality (\ref{eq:base1}) holds because, by the induction hypothesis,
$C(t)+2t+1=C(1+(t-1))+(t-1)+t+2 \geq C(2+(t-1))+t+2 > C(t+1)$.

\item If $s\notin2\NN$, say $s=2t+1$ with $t\geq 1$, then
$C(1+s) + s=C(2t+2)+2t+1=2C(t+1)+2t+1$ and
$C(2+s)=C(2t+3)=C(t+2)+C(t+1)+1$ and the desired inequality (\ref{eq:base1}) holds because, by the induction hypothesis,
$C(t+1)+2t\geq C(t+2)+t> C(t+2)$.
\end{itemize}

This completes the proof of the base case $n=1$.
Let us consider now the case $n\geq 2$ and let us assume that $C(n' + s) + C(n') + s \geq C(2n' + s)$ for every $1\leq n' < n$ and $s\in\NN$. To prove that (\ref{eq:goal}) is true for every $s\in\NN$ we distinguish 4 cases.

\begin{itemize}
\item If $n\in 2\NN$ and $s\in 2\NN$, say, $n=2m$ and $s=2t$, then
$$
\begin{array}{l}
C(n+s) +C(n)+ s=C(2m+2t)+C(2m)+2t=2(C(m+t)+C(m)+t)\\ 
C(2n+s)=C(4m+2t)=2C(2m+t)
\end{array}
$$
and the desired inequality (\ref{eq:goal}) is true because, by induction,
$C(m+t)+C(m)+t\geq C(2m+t)$.

\item If $n\in2\NN$ and $s\notin2\NN$, say, $n=2m$ and $s=2t+1$, then
$$
\begin{array}{l}
C(n+s) +C(n)+ s=C(2m+2t+1)+C(2m)+2t+1\\
\qquad\qquad=C(m+t+1)+C(m+t)+1+2C(m)+2t+1\\ 
C(2n+s)=C(4m+2t+1)=C(2m+t+1)+C(2m+t)+1
\end{array}
$$
and  (\ref{eq:goal}) holds because, by induction, 
$C(m+t+1)+C(m)+t+1\geq C(2m+t+1)$ and $C(m+t)+C(m)+t \geq C(2m+t)$

\item If $n\notin2\NN$ and $s\notin2\NN$, say, $n=2m+1$ and $s=2t+1$, then
$$
\begin{array}{l}
C(n+s) +C(n)+ s=C(2m+2t+2)+C(2m+1)+2t+1\\
\qquad\qquad=2C(m+t+1)+C(m+1)+C(m)+1+2t+1\\ 
C(2n+s)=C(4m+2t+3)=C(2m+t+2)+C(2m+t+1)+1
\end{array}
$$
and (\ref{eq:goal}) is true because, by induction, 
$C(m+t+1)+C(m+1)+t \geq C(2m+2+t)$ and
$C(m+t+1)+C(m)+t+1\geq C(2m+t+1)$.

\item Assume finally that $n\notin2\NN$ and $s\in2\NN$: say, $n=2m+1$ and $s=2t$. 
If $t=0$, then the desired inequality (\ref{eq:goal}) amounts to $C(n)+C(n)\geq C(2n)$, which is true because it is actually an equality. So, assume that $t\geq 1$.
In this case, 
$$
\begin{array}{l}
C(n+s) +C(n)+ s=C(2m+2t+1)+C(2m+1)+2t\\
\qquad\qquad=C(m+t+1)+C(m+t)+1+C(m+1)+C(m)+1+2t\\ 
C(2n+s)=C(4m+2+2t)=2C(2m+t+1)
\end{array}
$$
and the (\ref{eq:goal}) holds because, by induction, 
$C(m+t+1)+C(m)+t+1\geq C(2m+t+1)$
and
$$
\begin{array}{l}
C(m+t)+C(m+1)+t+1 =C((m+1)+(t-1))+C(m+1)+t-1+2\\ \qquad\qquad \geq C(2(m+1)+t-1)+2>C(2m+t+1)
\end{array}
$$
\end{itemize}
This completes the proof of the inductive step.
\end{proof}

\begin{rem}\label{rem:postkey}
Notice that in the proof of Lemma \ref{lem:key} we have proved that the following two facts:
\begin{enumerate}[(a)]
\item $C(1+s)+C(1)+s=C(2+s)$ if, and only if, $s\leq 1$.
\item If $s\geq 2$ is even and $n$ is odd, then $C(n+s)+C(n)+s>C(2n+s)$.
\end{enumerate}
\end{rem}

\begin{Thm}\label{thm:minC}
For every $n\geq 1$, the minimum Colless index on $\TT_n$ is $C(n)$.
\end{Thm}

\begin{proof}
We shall prove  by induction on $n$ that $C(T)\geq C(n)$ for every $T\in \TT_n$. The case when $n=1$ is obvious, because $\TT_1=\{B_1\}$. Assume now that the assertion is true for every number of leaves smaller than $n$ and let $T\in \TT_n$. Let $T_1\in\TT_{n_1}$ and $T_2\in \TT_{n_2}$ be its subtrees rooted at the children of its root and assume that $n_1\geq n_2$. Then
$$
C(T) = C(T_1) + C(T_2) + n_1-n_2 \geq  C(n_1) + C(n_2) + n_1-n_2 \geq C(n_1+n_2)=C(n)
$$
where the first inequality holds by the induction hypothesis and the second inequality by the previous lemma (taking in it $s=n_1-n_2$).
\end{proof}

This theorem says that the minimum Colless index on $\TT_n$ is reached at the maximally balanced bifurcating trees. When $n$ is a power of 2, Lemma \ref{lem:C=0} guarantees that this minimum is reached exactly at these trees, which in this case are the fully symmetric bifurcating trees. But for arbitrary values of $n$ there may be other trees whose Colles index is $C(n)$. For instance, the minimum Colless index on $\TT_6$, $C(6)=2$, is reached at the trees $B_6$ and $B_2\star B_4$ depicted in Figure \ref{fig:min6}. 

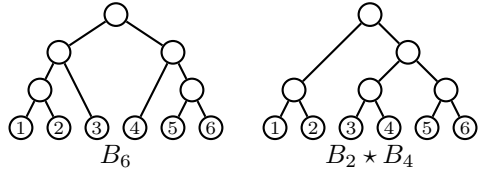
\begin{figure}[htb]
\begin{center}
\begin{tikzpicture}[thick,>=stealth,scale=0.25]
\draw(0,0) node [tre] (1) {};  \etq 1
\draw(2,0) node [tre] (2) {};  \etq 2
\draw(4,0) node [tre] (3) {};  \etq 3
\draw(6,0) node [tre] (4) {};  \etq 4
\draw(8,0) node [tre] (5) {};  \etq 5
\draw(10,0) node [tre] (6) {};  \etq 6
\draw(1,2) node[tre] (a) {};
\draw(9,2) node[tre] (c) {};
\draw(2,4) node[tre] (b) {};
\draw(8,4) node[tre] (d) {};
\draw(5,6) node[tre] (r) {};
\draw  (a)--(1);
\draw  (a)--(2);
\draw  (b)--(a);
\draw  (b)--(3);
\draw  (c)--(5);
\draw  (c)--(6);
\draw  (d)--(4);
\draw  (d)--(c);
\draw  (r)--(b);
\draw  (r)--(d);
\draw (5,-1.5) node {$B_6$};
\end{tikzpicture}
\quad
\begin{tikzpicture}[thick,>=stealth,scale=0.25]
\draw(0,0) node [tre] (1) {};  \etq 1
\draw(2,0) node [tre] (2) {};  \etq 2
\draw(4,0) node [tre] (3) {};  \etq 3
\draw(6,0) node [tre] (4) {};  \etq 4
\draw(8,0) node [tre] (5) {};  \etq 5
\draw(10,0) node [tre] (6) {};  \etq 6
\draw(1,2) node[tre] (a) {};
\draw(5,2) node[tre] (b) {};
\draw(9,2) node[tre] (c) {};
\draw(7,4) node[tre] (d) {};
\draw(5,6) node[tre] (r) {};
\draw  (a)--(1);
\draw  (a)--(2);
\draw  (b)--(4);
\draw  (b)--(3);
\draw  (c)--(5);
\draw  (c)--(6);
\draw  (d)--(b);
\draw  (d)--(c);
\draw  (r)--(a);
\draw  (r)--(d);
\draw (5,-1.5) node {$B_2\star B_4$};
\end{tikzpicture}
\end{center}
\caption{\label{fig:min6} 
The trees in $\TT_6$ with minimum Colless index.}
\end{figure}

\section{Which trees have Colless index $C(n)$?}\label{sec:QB}

\noindent In this section we provide a way of generating all bifurcating trees with minimum Colless index among all bifurcating trees with their number of leaves.

\begin{lemma}\label{lem:charmin1}
Let $T\in \TT_n$. Then, $C(T)=C(n)$ if, and only if, for every $v\in V_{int}(T)$, say with children $v_1,v_2$, 
$C(\kappa_T(v_1))+C(\kappa_T(v_2))+|\kappa_T(v_1)-\kappa_T(v_2)|=C(\kappa_T(v))$.
\end{lemma}

\begin{proof}
$\Longrightarrow$) Assume that there exists some $v\in V_{int}(T)$, with children $v_1,v_2$, such that
$C(\kappa_T(v_1))+C(\kappa_T(v_2))+|\kappa_T(v_1)-\kappa_T(v_2)|\neq C(\kappa_T(v))$. By Lemma \ref{lem:key}, this implies that 
$C(\kappa_T(v_1))+C(\kappa_T(v_2))+|\kappa_T(v_1)-\kappa_T(v_2)|> C(\kappa_T(v))$. Let
$T'\in \TT_n$ be the tree obtained by replacing in $T$ the rooted subtree $T_v$ by a maximally balanced tree $B_{\kappa_T(v)}$ and leaving the rest of $T$ untouched. In this way, $T'_v=B_{\kappa_T(v)}$ and $bal_T(x)=bal_{T'}(x)$ for every  $x\in V_{int}(T)\setminus V_{int}(T_v)=V_{int}(T')\setminus V_{int}(T'_v)$; let us denote by $W$ this last set of nodes. Then
\begin{align*}
C(T) & =\sum_{x\in W} bal_T(x)+C(T_v) = \sum_{x\in W} bal_T(x)+C(T_{v_1})+C(T_{v_2})+|\kappa_T(v_1)-\kappa_T(v_2)|\\
 & \geq \sum_{x\in W} bal_T(x)+C(\kappa_T(v_1))+C(\kappa_T(v_2))+|\kappa_T(v_1)-\kappa_T(v_2)|\\
 & >\sum_{x\in W} bal_T(x)+C(\kappa_T(v))=\sum_{x\in W} bal_{T'}(x)+C(T'_v)=C(T')\geq C(n)
\end{align*}
 This proves the ``only if'' implication. 
 
\noindent $\Longleftarrow$)  We prove the ``if'' implication by induction on $n$. The case when $n=1$ is obvious, because $\TT_1=\{B_1\}$. Assume now that this implication is true for every tree in $\TT_{n'}$ with $n'<n$, and let $T\in \TT_n$
be such that, for every $v\in V_{int}(T)$, 
$C(\kappa_T(v_1))+C(\kappa_T(v_2))+|\kappa_T(v_1)-\kappa_T(v_2)|=C(\kappa_T(v))$, where $v_1,v_2$ stand for the children of $v$. 
Let $r$ be the root of $T$ and $x_1,x_2$ its two children.
Then, for every $v\in V_{int}(T_{x_1})$, with children $v_1,v_2$, 
$$
\begin{array}{l}
C(\kappa_{T_{x_1}}(v_1))+C(\kappa_{T_{x_1}}(v_2))+|\kappa_{T_{x_1}}(v_1)-\kappa_{T_{x_1}}(v_2)|\\
\quad =C(\kappa_{T}(v_1))+C(\kappa_{T}(v_2))+|\kappa_{T}(v_1)-\kappa_{T}(v_2)|=C(\kappa_{T}(v))=C(\kappa_{T_{x_1}}(v)).
\end{array}
$$
This implies, by the induction hypothesis, that $C(T_{x_1})=C(\kappa_T(x_1))$. 
By symmetry, we also have that $C(T_{x_2})=C(\kappa_T(x_2))$.
Finally,  
\begin{align*}
C(T) & =C(T_{x_1})+C(T_{x_2})+|\kappa_T(x_1)-\kappa_T(x_2)|\\ & =
C(\kappa_T(x_1))+C(\kappa_T(x_2))+|\kappa_T(x_1)-\kappa_T(x_2)| =C(\kappa_T(r))=C(n)
\end{align*}
as we wanted to prove.
\end{proof}

This lemma, together with Lemma \ref{lem:C=0}, provide the following algorithm to produce all trees $T\in \TT_n$ such that $C(T)=C(n)$. In it, and henceforth, for every $n\geq 2$, let
$$
QB(n)=\big\{ (n_1,n_2)\in \NN^2\mid  1\leq n_2\leq n_1,\ n_1+n_2=n, C(n_1)+C(n_2)+n_1-n_2=C(n)\big\}.
$$ 
Notice that $QB(n)\neq \emptyset$, because $(\ceil*{n/2},\floor*{n/2})\in QB(n)$.\medskip

\noindent\textbf{Algorithm 1.}

\begin{enumerate}
\item[1)] Start with a single node labeled $n$.
\item[2)] While  the current tree contains labeled leaves, choose a leaf with label $m$.
\begin{itemize}
\item[2.a)] If $m$ is a power of 2, replace this leaf by a fully bifurcating tree $B_m$ with its leaves unlabeled.
\item[2.b)] If $m$ is a not a power of 2:
\begin{itemize}
\item[2.b.i)] Find a pair of integers $(m_1,m_2)\in QB(m)$.
\item[2.b.ii)] Split the leaf labeled $m$ into a cherry with leaves labeled $m_1$ and $m_2$, respectively.
\end{itemize}
\end{itemize}
\end{enumerate}

Next result provides  a characterization of the pairs $(n_1,n_2)\in QB(n)$  which will allow to perform efficiently step (2.b.i) in this algorithm.

\begin{prop}\label{prop:eqC}
For every $n\geq 2$ and for every $1\leq n_2\leq n_1$ such that $n_1+n_2=n$:
\begin{enumerate}[(1)]

\item If $n_1=n_2=n/2$, then $(n_1,n_2)\in QB(n)$ always.

\item If $n_1>n_2$, then $(n_1,n_2)\in QB(n)$ if, and only if, there exist $k\geq 0$, $l\geq 1$, $p\geq 1$, and  $0\leq t<2^{l-2}$ such that one of the following three conditions holds:
\begin{itemize}
\item There exist $k\geq 0$ and $p\geq 1$ such that $n=2^k(2p+1)$, $n_1=2^k(p+1)$ and $n_2=2^kp$.

\item There exist $k\geq 0$, $l\geq 2$, $p\geq 1$, and  $0\leq t<2^{l-2}$ such that $n=2^k(2^l(2p+1)-(2t+1))$, $n_1=2^k(2^l(p+1)-(2t+1))$, and $n_2=2^{k+l}p$.

\item There exist $k\geq 0$, $l\geq 2$, $p\geq 1$, and  $0\leq t<2^{l-2}$ such that $n=2^k(2^l(2p+1)+2t+1)$, $n_1=2^{k+l}(p+1)$, and $n_2=2^{k}(2^lp+2t+1)$. 
\end{itemize}
\end{enumerate}
\end{prop}

We split the proof of this result into several auxiliary lemmas.

\begin{lemma}\label{lem:eqeven}
Let $s=2^kt$ with $k\geq 1$ and $t\geq 1$. Then, for every $n\geq 1$, $(n+s,n)\in QB(2n+s)$  if, and only if,  $n=2^km$, for some $m\geq 1$ such that $(m+t,m)\in QB(2m+t)$.
\end{lemma}

\begin{proof}
We prove the equivalence in the statement by induction on the exponent $k\geq 1$.
Recall that, by Remark \ref{rem:postkey}.(b), if $s\geq 1$ is even and $C(n+s)+C(n)+s=C(2n+s)$, then $n$ must be even, too.  Therefore, if $s=2t_0$, then $n=2m_0$ for some $m_0\geq 1$, and then, since
$C(2m_0+2t_0)+C(2m_0)+2t_0=2(C(m_0+t_0)+C(m_0)+t_0)$ and 
$C(4m_0+2t_0)=2C(2m_0+t_0)$,
the equality $C(n+s)+C(n)+s=C(2n+s)$ is equivalent to the equality $C(m_0+t_0)+C(m_0)+t_0=C(2m_0+t_0)$. 
This proves the equivalence in the statement when $k=1$. 

Now, assume that this equivalence is true  for the exponent $k-1$, and let $s=2^kt$. Then, by the case $k=1$, $C(n+s)+C(n)+s=C(2n+s)$ if, and only if, 
$n=2m_0$ for some $m_0\geq 1$ such that $C(m_0+2^{k-1}t)+C(m_0)+2^{k-1}t=C(2m_0+2^{k-1}t)$, and, by the induction hypothesis, this last equality holds if, and only if, 
$m_0=2^{k-1}m$ for some $m\geq 1$ such that $C(m+t)+C(m)+t=C(2m+t)$. Combining both equivalences we obtain the equivalence in the statement, thus proving the inductive step.
\end{proof}

\begin{lemma}\label{lem:eqodd2}
Let  $s=2^{k+1}-(2t+1)$, with $k=\floor*{\log_2(s)}$ and  $0\leq t<2^{k-1}$, be an odd integer. Then, for every $m\geq 1$, 
$(2m+s,2m)\in QB(4m+s)$ if, and only if, $m=2^{k}p$ for some $p\geq 1$.
\end{lemma}

\begin{proof}
We prove the equivalence in the statement by induction on $s$. When $s=1=2^{1}-1$, the equivalence says that
$C(2m+1)+C(2m)+1=C(4m+1)$  for every $m\geq 1$, which is true by (\ref{eq:recC}).

Assume now that the equivalence is true for every odd natural number $s'<s$ and for every $m$, and let us prove it for $s=2^{k+1}-(2t+1)$ with $0\leq t<2^{k-1}$. We have that
$$
\begin{array}{l}
C(2m+2^{k+1}-2t-1)+C(2m)+2^{k+1}-2t-1\\
\qquad =\big(C(m+2^k-t)+C(m)+2^k-t\big)+\big(C(m+2^k-t-1)+C(m)+2^k-t-1\big)+1\\
C(4m+2^{k+1}-2t-1)=C(2m+2^k-t)+C(2m+2^k-t-1)+1
 \end{array}
 $$
 and since, by Lemma \ref{lem:key},
 $C(m+2^k-t)+C(m)+2^k-t\geq C(2m+2^k-t)$ and $C(m+2^k-t-1)+C(m)+2^k-t-1\geq C(2m+2^k-t-1)$,
 we have that $C(2m+s)+C(2m)+s=C(4m+s)$ if, and only if, the following two identities are satisfied:
\begin{align}
& C(m+2^k-t)+C(m)+2^k-t= C(2m+2^k-t) \label{eq:eqodd2-1}\\
& C(m+2^k-t-1)+C(m)+2^k-t-1= C(2m+2^k-t-1)  \label{eq:eqodd2-2}
\end{align}
So, we must prove that (\ref{eq:eqodd2-1}) and (\ref{eq:eqodd2-2}) hold if, and only if, $m=2^{k}p$ for some $p\geq 1$.
We  distinguish two subcases, depending on the parity of $t$:
\begin{itemize}

\item If $t=2x$ for some $0\leq x<2^{k-2}$, then (\ref{eq:eqodd2-1}) and Lemma  \ref{lem:eqeven} imply that $m$ is even, say $m=2m_0$, and then  (\ref{eq:eqodd2-2}) says
\begin{equation}
C(2m_0+2^k-2x-1)+C(2m_0)+2^k-2x-1= C(4m_0+2^k-2x-1),
\label{eq:eqodd2-3}
\end{equation}
which, by induction,  is equivalent to $m_0=2^{k-1}p$ for some $p\geq 1$, i.~e., to $m=2^{k}p$ for some $p\geq 1$.
So, to complete the proof of the desired equivalence, it remains to prove that if $m=2^{k}p$, then
(\ref{eq:eqodd2-1}) holds. If $t=0$, this equality is a direct consequence  of Lemma   \ref{lem:eqeven} and   (\ref{eq:recC}),  so assume that $t>0$ and write it as $t=2^j(2x_0+1)$ with $1\leq j<k-1$ and $x_0<2^{k-j-2}$. Then
$$
\begin{array}{l}
C(m+2^k-t)+C(m)+2^k-t  = C(2^{k}p+2^k-2^j(2x_0+1))+C(2^{k}p)+2^k-2^j(2x_0+1)\\
\qquad = 2^j\big(C(2^{k-j}p+2^{k-j}-2x_0-1)+C(2^{k-j}p)+2^{k-j}-2x_0-1\big)\\
\qquad = 2^jC(2^{k-j+1}p+2^{k-j}-2x_0-1) \mbox{ (by the induction hypothesis)}\\
\qquad = C(2^{k+1}p+2^{k}-2^j(2x_0+1))=C(2m+2^k-t)
\end{array}
$$

\item If $t=2x+1$ for some $0\leq x< 2^{k-2}$, then (\ref{eq:eqodd2-2}) and Lemma  \ref{lem:eqeven} imply that $m$ is even, say $m=2m_0$, and then it is (\ref{eq:eqodd2-1}) which  is equivalent to equation (\ref{eq:eqodd2-3}) above, which, on its turn, by induction
is equivalent to $m_0=2^{k-1}p$ for some $p\geq 1$, that is, to $m=2^{k}p$ for some $p\geq 1$
Thus, to complete the proof of the desired equivalence, it remains to prove that if $m=2^{k}p$, then
(\ref{eq:eqodd2-2}) holds. Now:
$$
\begin{array}{l}
C(m+2^k-t-1)+C(m)+2^k-t-1  = C(2^{k}p+2^k-2x-2)+C(2^{k}p)+2^k-2x-2\\
\qquad =2\big(C(2^{k-1}p+2^{k-1}-x-1)+C(2^{k-1}p)+2^{k-1}-x-1\big)
\end{array}
$$
Now, if $x$ is even, say $x=2x_0$, then, since $x_0<2^{k-3}$, the induction hypothesis implies that
$$
\begin{array}{l}
2\big(C(2^{k-1}p+2^{k-1}-x-1)+C(2^{k-1}p)+2^{k-1}-x-1\big) =2 C(2^{k}p+2^{k-1}-x-1) \\
\qquad = C(2^{k+1}p+2^{k}-2x-2)= C(2m+2^k-t-1)
\end{array}
$$
And if $x$ is odd, write it as $x=2^j(2t_0+1)-1$ for some $1\leq j<k-1$ (and notice that $x<2^{k-2}$ implies $t_0<2^{k-j-3}$) and then
$$
\begin{array}{l}
2\big(C(2^{k-1}p+2^{k-1}-x-1)+C(2^{k-1}p)+2^{k-1}-x-1\big)\\
\qquad=
2\big(C(2^{k-1}p+2^{k-1}-2^j(2t_0+1))+C(2^{k-1}p)+2^{k-1}-2^j(2t_0+1)\big)\\
\qquad =2\cdot 2^j\big(C(2^{k-j-1}p+2^{k-j-1}-(2t_0+1))+C(2^{k-j-1}p)+2^{k-j-1}-(2t_0+1)\big)\\
\qquad = 2^{j+1} C(2^{k-j}p+2^{k-j-1}-(2t_0+1)) \mbox{ (by the induction hypothesis)}\\
\qquad = C(2^{k+1}p+2^{k}-2^{j+1}(2t_0+1))=C(2^{k+1}p+2^{k}-2x-2)=C(2m+2^k-t-1)
\end{array}
$$
This completes the proof of the desired equivalence when $t$ is odd.
\end{itemize}
So, the inductive step is true in all cases.
\end{proof}

\begin{lemma}\label{lem:eqodd3}
Let  $s=2^{k+1}-(2t+1)$, with $k=\floor*{\log_2(s)}$ and  $0\leq t<2^{k-1}$, be an odd integer.  Then, for every $m\geq 0$, $(2m+1+s,2m+1)\in QB(4m+2+s)$ if, and only if, either $m=2^{k}p+t$ for some $p\geq 1$ or $s=1$ (i.e., $k=t=0$) and $m=0$.
\end{lemma}

\begin{proof}
We prove the equivalence in the statement by induction on $s$. When $s=1=2^{1}-1$, the equivalence says that
$C(2m+2)+C(2m+1)+1=C(4m+3)$  for every $m\geq 0$, which is true by (\ref{eq:recC}).

Assume now that the equivalence is true for every odd natural number $s'<s$ and for every $m\geq 0$, and let us prove it for $s=2^{k+1}-(2t+1)$ with $0\leq t<2^{k-1}$. In this case, if $m=0$, then by Remark \ref{rem:postkey}.(a) we know that $(s+1,s)\in QB(s+2)$ if, and only if, $s=1$. So, assume that $m\geq 1$. Then, we have that
$$
\begin{array}{l}
C(2m+1+2^{k+1}-2t-1)+C(2m+1)+2^{k+1}-2t-1\\
\qquad =\big(C(m+2^k-t)+C(m)+2^k-t\big)+\big(C(m+2^k-t)+C(m+1)+2^k-t-1\big)+1\\
C(4m+2+2^{k+1}-2t-1)=C(2m+2^k-t)+C(2m+2^k-t+1)+1
 \end{array}
 $$
 and since, by Lemma \ref{lem:key},
 $C(m+2^k-t)+C(m)+2^k-t\geq C(2m+2^k-t)$ and 
$C(m+2^k-t)+C(m+1)+2^k-t-1\geq C(2m+2^k-t+1)$,
 we have that $C(2m+1+s)+C(2m+1)+s=C(4m+2+s)$ if, and only if,
\begin{align}
& C(m+2^k-t)+C(m)+2^k-t= C(2m+2^k-t) \label{eq:eqodd3-1}\\
& C(m+2^k-t)+C(m+1)+2^k-t-1= C(2m+2^k-t+1) \label{eq:eqodd3-2}
\end{align}
So, we must prove that (\ref{eq:eqodd3-1}) and (\ref{eq:eqodd3-2}) hold for $m\geq 1$ if, and only if, $m=2^{k}p+t$ for some $p\geq 1$. We distinguish again two subcases, depending on the parity of $t$:

\begin{itemize}

\item If $t=2x$ for some $0\leq x<2^{k-2}$, then (\ref{eq:eqodd3-1}) and Lemma  \ref{lem:eqeven} imply that $m$ is even, say $m=2m_0$ with $m_0\geq 1$, and then (\ref{eq:eqodd3-2}) can be written 
$$
C(2m_0+1+2^k-2x-1)+C(2m_0+1)+2^k-2x-1= C(4m_0+2+2^k-2x-1)
$$ 
which, by induction,  is equivalent to $m_0=2^{k-1}p+x$ for some $p\geq 1$, that is, to $m=2^{k}p+t$ for some $p\geq 1$. So, to complete the proof of the desired equivalence, it remains to check that if $m=2^{k}p+t$, then (\ref{eq:eqodd3-1}) holds. Now, if $x=0$, so that $m=2^{k}p$,  (\ref{eq:recC}) and Lemma   \ref{lem:eqeven} clearly imply (\ref{eq:eqodd3-1}). So, assume that $x>0$ and write it as $x=2^j(2y_0+1)$ with $0\leq j< k-2$ and $y_0<2^{k-j-3}$. Then
$$
\begin{array}{l}
C(m+2^k-t)+C(m)+2^k-t=C(2^{k}p+2x+2^k-2x)+C(2^{k}p+2x)+2^k-2x\\
\quad= C(2^{k}p+2^{j+1}(2y_0+1)+2^{k}-2^{j+1}(2y_0+1))+C(2^{k}p+2^{j+1}(2y_0+1))+2^{k}-2^{j+1}(2y_0+1)\\
\quad =2^{j+1}\big(C(2^{k-j-1}p+2y_0+1+2^{k-j-1}-(2y_0+1))+C(2^{k-j-1}p+2y_0+1)+2^{k-j-1}-(2y_0+1)\big)\\
\quad =2^{j+1} C(2^{k-j}p+4y_0+2+2^{k-j-1}-(2y_0+1))\mbox{ (by the induction hypothesis)}\\
\quad =C(2^{k+1}p+2^{k}+2^{j+1}(2y_0+1))=C(2^{k+1}p+2^{k}+2x)=C(2m+2^k-t)
\end{array}
$$
as we wanted to prove.

\item If $t=2x+1$ for some $0\leq x<2^{k-2}$,  (\ref{eq:eqodd3-2}) and Lemma  \ref{lem:eqeven} imply that $m+1$ is even, and then $m$ is odd, say $m=2m_0+1$ for some $m_0\geq 0$, and  (\ref{eq:eqodd3-1}) can be written 
\begin{equation}
C((2m_0+1)+2^k-2x-1)+C(2m_0+1)+2^k-2x-1= C(4m_0+2+2^k-2x-1).
\label{eq:encaraunaltra}
\end{equation}
Now, if $m_0=0$, Remark \ref{rem:postkey}.(a) implies that this equality holds if, and only if, $2^k-2x-1=1$ which, under the condition $0\leq x<2^{k-2}$, only happens when $k=1$ and $x=0$, but then $t=1=2^{k-1}$ against the assumption that  $t<2^{k-1}$. Therefore $m_0$ must be at least 1.

Then, by induction,  identity (\ref{eq:encaraunaltra})
 is equivalent to $m_0=2^{k-1}p+x$ for some $p\geq 1$, that is, to $m=2^{k}p+t$ for some $p\geq 1$.  So, to complete the proof of the desired equivalence, it remains to check that if $m=2^{k}p+t$, then (\ref{eq:eqodd3-2}) holds.  Now, in the current situation:
$$
\begin{array}{l}
C(m+2^k-t)+C(m+1)+2^k-t-1\\
\qquad =C(2^{k}p+2x+1+2^k-2x-1)+C(2^{k}p+2x+2)+2^k-2x-2\\
\qquad = 2\big(C((2^{k-1}p+x+1)+(2^{k-1}-x-1))+C(2^{k-1}p+x+1)+2^{k-1}-x-1\big)=(**)
\end{array}
$$

If $x$ is even, say $x=2x_0$ with $0\leq x_0<2^{k-3}$, then 
\begin{align*}
(**) & =2\big(C((2^{k-1}p+2x_0+1)+(2^{k-1}-2x_0-1))+C(2^{k-1}p+2x_0+1)+2^{k-1}-2x_0-1\big)\\ 
& = 2C(2^{k}p+2(2x_0+1)+2^{k-1}-(2x_0+1))\mbox{ (by induction)}\\
&=C(2^{k+1}p+2^{k}+4x_0+2) =C(2m+2^k-t+1)
\end{align*}

And if $x$ is odd, write it as $x=2^j(2t_0+1)-1$ with $1\leq j<k-1$ and $t_0<2^{k-j-3}$, and then
\begin{align*}
(**) & =2\big(C(2^{k-1}p+2^j(2t_0+1)+2^{k-1}-2^j(2t_0+1))+C(2^{k-1}p+2^j(2t_0+1))+2^{k-1}-2^j(2t_0+1)\big)\\
 & = 2^{j+1}\big(C(2^{k-j-1}p+2t_0+1+2^{k-j-1}-(2t_0+1))+C(2^{k-j-1}p+2t_0+1)+2^{k-j-1}-(2t_0+1)\big)\\
& = 2^{j+1}C(2^{k-j}p+4t_0+2+2^{k-j-1}-(2t_0+1)) \mbox{ (by the induction hypothesis)}\\
& = C(2^{k+1}p+2^{j+1}(2t_0+1)+2^{k})=C(2^{k+1}p+2x+2+2^{k})=C(2m+2^k-t+1)
\end{align*}
This completes the proof of the desired equivalence when $t$ is odd. 
\end{itemize}
\end{proof}

We can return now to Proposition \ref{prop:eqC}.\medskip

\noindent\textit{Proof of Proposition \ref{prop:eqC}.} Assertion (1) is a direct consequence of identity (\ref{eq:recC}). So, assume $n_1>n_2$ and set $s=n_1-n_2$, so that $n_1=n_2+s$. Then:
\begin{enumerate}[(a)]
\item If $s=1$, then, by Lemma \ref{lem:eqodd3}, $C(n_1)+C(n_2)+n_1-n_2=C(n_1+n_2)$ for every $n_2\geq 1$.

\item If $s>1$ is odd, write it as $s=2^{j+1}-(2t+1)$, with $j=\floor*{\log_2(s)}\geq 1$ and  $0\leq t<2^{j-1}$. Then, by Lemmas \ref{lem:eqodd2} and \ref{lem:eqodd3}, $C(n_1)+C(n_2)+n_1-n_2=C(n_1+n_2)$ if, and only if, either
$n_2=2^{j+1}p$ or $n_2=2^{j+1}p+2t+1$, for some $p\geq 1$.

\item If $s\geq 2$ is even, write it as $s=2^ks_0$, with $k\geq 1$ the largest exponent of a power of 2 that divides $s$ and $s_0$  an odd integer, which we write $s_0=2^{j+1}-(2t+1)$ with $j=\floor*{\log_2(s_0)}$ and  $0\leq t<2^{j-1}$. Then, by Lemma \ref{lem:eqeven}, $C(n_1)+C(n_2)+n_1-n_2=C(n_1+n_2)$ if, and only if,
$n_2=2^km$, for some $m\geq 1$ such that $C(m+s_0)+C(m)+s_0=C(2m+s_0)$, and then:
\begin{itemize}
\item If $s_0=1$, $C(m+s_0)+C(m)+s_0=C(2m+s_0)$ for every $m\geq 1$ and therefore, in this case, $C(n_1)+C(n_2)+n_1-n_2=C(n_1+n_2)$ for every $n_2=2^km$ with $m\geq 1$.
\item If $s_0>1$, Lemmas \ref{lem:eqodd2} and \ref{lem:eqodd3} imply that $C(m+s_0)+C(m)+s_0=C(2m+s_0)$ if, and only if, $m=2^{j+1}p$ or $m=2^{j+1}p+2t+1$, for some $p\geq 1$. Therefore, in this case, $C(n_1)+C(n_2)+n_1-n_2=C(n_1+n_2)$ if, and only if, 
$n_2=2^{k+j+1}p$ or $n_2=2^k(2^{j+1}p+2t+1)$, for some $p\geq 1$.
\end{itemize}
\end{enumerate}

Combining the three cases, and taking $k=0$ in the odd $s$ case, we conclude that $C(n_1)+C(n_2)+n_1-n_2=C(n_1+n_2)$ if, and only if, $n_1-n_2=2^k(2^{j+1}-(2t+1))$ (for some $k\geq 0$, $j\geq 0$, and $0\leq t<2^{j-1}$) and
\begin{itemize} 
\item If $j=t=0$, then $n_2=2^kp$ for some $p\geq 1$, in which case $n_1=2^kp+1$ and $n=2^{k+1}p+1$

\item If $j>0$ or $t>0$, then one of the following two conditions holds for some $p\geq 1$:
\begin{itemize}
\item $n_2=2^{k+j+1}p$,  in which case $n_1=2^k(2^{j+1}(p+1)-(2t+1))$ and $n=2^k(2^{j+1}(2p+1)-(2t+1))$; or 

\item  $n_2=2^k(2^{j+1}p+2t+1)$, $n_1=2^{k+j+1}(p+1)$ and $n=2^k(2^{j+1}(2p+1)+ 2t+1)$
\end{itemize}
\end{itemize}
This is equivalent to the expressions for $n_1$ and $n_2$ in option (2) in the statement (replacing $j+1$ by $l\geq 1$). \qed

\begin{prop}\label{cor:QB}
For every $n\geq 2$, let $k\geq 0$ be the exponent of the largest power of 2 that divides $n$, let $n_0=n/2^k$, and let $n_0=\sum_{i=1}^\ell 2^{m_i}$, with $m_1>\cdots>m_{\ell-1}>m_\ell=0$, be the binary expansion of $n_0$. Then
\begin{enumerate}[(a)]
\item If $\ell=1$, i.e., if $n=2^k$, then $QB(n)=\{(n/2,n/2)\}$.

\item If $\ell>1$:
\begin{itemize}
\item[(b.1)] $QB(n)$ always contains the pair
$$
\Big(2^k\Big(\sum_{i=1}^{\ell-1} 2^{m_i-1}+1\Big), 2^k\sum_{i=1}^{\ell-1} 2^{m_i-1}\Big).
$$

\item[(b.2)] For every $j=2,\ldots, \ell-1$ such that $m_j>m_{j+1}+1$, $QB(n)$ contains the pair 
$$
\Big(2^{k}\Big(\sum_{i=1}^{j-1} 2^{m_i-1}+2^{m_j}\Big),n-2^{k}\Big(\sum_{i=1}^{j-1} 2^{m_i-1}+2^{m_j}\Big)\Big).
$$

\item[(b.3)] For every $j=2,\ldots,\ell-1$ such that $m_j<m_{j-1}-1$, 
$QB(n)$ contains the pair 
$$
\Big(n-2^{k}\sum_{i=1}^{j-1} 2^{m_i-1},2^{k}\sum_{i=1}^{j-1} 2^{m_i-1}\Big).
$$

\item[(b.4)] If $k\geq 1$, then $QB(n)$ contains the pair $(n/2,n/2)$.
\end{itemize}
The pairs  described in  (b.1) to (b.4) are pairwise different, and $QB(n)$ contains no other member.
\end{enumerate}
\end{prop}

\begin{proof}
Assertion (a) is obvious by Lemma \ref{lem:C=0}. So, assume henceforth that $\ell>1$. Let now $(n_1,n_2)$ such that $n=n_1+n_2$ and $1\leq n_2< n_1$. Then, by Proposition \ref{prop:eqC}, $(n_1,n_2)\in QB(n)$ if, and only if,  one of the following three conditions is satisfied:
\begin{itemize}
\item[(b.1)] There exist $k\geq 0$ and $p\geq 1$ such that $n_0=2p+1$, $n_1=2^k(p+1)$, and $n_2=2^{k}p$.
In this case 
$p=(n_0-1)/{2}=\sum_{i=1}^{\ell-1} 2^{m_i-1}$
 and this contributes to $QB(n)$ the pair
$(n_1,n_2)$ with
$$
n_1=2^{k}\Big(\sum_{i=1}^{\ell-1} 2^{m_i-1}+1\Big),\quad n_2=2^{k}\sum_{i=1}^{\ell-1} 2^{m_i-1}.
$$

\item[(b.2)] There exist $k\geq 0$, $l\geq 2$, $p\geq 1$, and  $0\leq t<2^{l-2}$ such that  $n_0=2^{l+1}p+2^l+2t+1$ and $n_1=2^{k+l}(p+1)$. Now, if $t<2^{l-2}$ and $p\geq 1$, then $2t+1<2^{l-1}$ and $2^{l+1}p\geq 2^{l+1}$. Therefore, the equality $$2^{l+1}p+2^l+2t+1=\sum_{i=1}^\ell 2^{m_i}$$ holds for some $p\geq 1$ and $t<2^{l-2}$ if, and only if,  $m_j=l\geq 2$ for some $j=2,\ldots,\ell-1$, in which case  $p=\big(\sum_{i=1}^{j-1} 2^{m_i}\big)/2^{m_j+1}$. This contributes to $QB(n)$ 
all pairs $(n_1,n_2)$ of the form 
\begin{equation}
\displaystyle n_1=2^{k+m_j}\Big(\frac{\sum_{i=1}^{j-1} 2^{m_i}}{2^{m_j+1}}+1\Big)=2^{k}\Big(\sum_{i=1}^{j-1} 2^{m_i-1}+2^{m_j}\Big),\
\displaystyle  n_2=n-2^{k}\Big(\sum_{i=1}^{j-1} 2^{m_i-1}+2^{m_j}\Big),
\label{eq:QB.b.2}
\end{equation}
with $j=2,\ldots,\ell-1$ and $m_j\geq 2$. But not all these pairs are different, because
$\sum_{i=1}^{j-1} 2^{m_i-1}+2^{m_j}=\sum_{i=1}^{j} 2^{m_i-1}+2^{m_{j+1}}$ if, and only if, $m_j=m_{j+1}+1$.
Indeed,
\begin{align*}
\sum_{i=1}^{j-1} 2^{m_i-1}+2^{m_j} & =\sum_{i=1}^{j} 2^{m_i-1}+2^{m_{j+1}}  \Longleftrightarrow 
 2^{m_j}=2^{m_j-1}+2^{m_{j+1}} \\
&\qquad  \Longleftrightarrow 2^{m_j-1}=2^{m_{j+1}}\Longleftrightarrow m_j=m_{j+1}+1.
\end{align*}
This entails that each sequence of values of consecutive exponents $m_j$ that fiffer only in 1  yield the same $n_1$, and hence the same pair $(n_1,n_2)$. On the other hand, $\sum_{i=1}^{j-1} 2^{m_i-1}+2^{m_j}$ is clearly monotonously non-increasing on $j$ (because $m_{j+1}< m_j$), and therefore, by the previous equivalence, its value jumps at each $j$ such that $m_j>m_{j+1}+1$. Finally, if it happens that $m_{\ell-2}=2$, then $m_{\ell-1}=1=m_{\ell-2}-1$ and $m_\ell=0=m_{\ell-1}-1$ and hence, as we have just seen,
$$
\sum_{i=1}^{\ell-3} 2^{m_i-1}+2^{m_{\ell-2}}=\sum_{i=1}^{\ell-2} 2^{m_i-1}+2^{m_{\ell-1}}=\sum_{i=1}^{\ell-1} 2^{m_i-1}+2^{m_{\ell}}=\sum_{i=1}^{\ell-1} 2^{m_i-1}+1
$$
and the value of $n_1$ that we obtain in this case, $2^k\big(\sum_{i=1}^{\ell-1} 2^{m_i-1}+1\big)$, is the one already given in (b.1), so we can omit it. In summary, the new pairwise different pairs $(n_1,n_2)$ of this form are obtained by taking in (\ref{eq:QB.b.2}) as $l$ the exponents $m_j$ with $j=2,\ldots,\ell-1$ such that $m_j>m_{j+1}+1$.

\item[(b.3)] There exist $k\geq 0$, $l\geq 2$, $p\geq 1$, and  $0\leq t<2^{l-2}$ such that $n_0=2^{l+1}p+2^l-(2t+1)$ and  $n_2=2^{k+l}p$. Since $t<2^{l-2}$, we have that $n_0=2^{l+1}p+2^{l-1}+2t_0+1$ with $2t_0+1<2^{l-1}$. Then, the equality $$2^{l+1}p+2^{l-1}+2t_0+1=\sum_{i=1}^\ell 2^{m_i}$$ holds for some $p\geq 1$ and $t_0<2^{l-2}$  if, and only if,  $l-1=m_j$ for some $j=2,\ldots,\ell-1$ such that $m_{j-1}>m_j+1$, and then  $$p=\frac{\sum_{i=1}^{j-1} 2^{m_i}}{2^{m_j+2}}.$$ This contributes to $QB(n)$  all pairs $(n_1,n_2)$ of the form 
$$
n_2=2^{k+m_j+1}\Big(\frac{\sum_{i=1}^{j-1} 2^{m_i}}{2^{m_j+2}}\Big)=2^k\sum_{i=1}^{j-1} 2^{m_i-1},\quad n_1=n-2^k\sum_{i=1}^{j-1} 2^{m_i-1},
$$
with $j=2,\ldots,\ell-1$ such that $m_j<m_{j-1}-1$, belong to $QB(n)$, and they are pairwise different. 
\end{itemize}
This gives all pairs $(n_1,n_2)$ in $QB(n)$ with $n_1>n_2$. If $n$ is even, we must add moreover to $QB(n)$ the pair $(n/2,n/2)$ and this completes the set of pairs belonging to $QB(n)$. To finish the proof of the statement, we must check that these pairs are pairwise different.

Now, along our construction we have already checked that the pairs of the form (b.2), as well as those of the form (b.3),  are pairwise different. We have also checked that the pairs of the form (b.2) and different from the pair (b.1). 
The pairs of the form (b.3)  are also different from the pair (b.1), because, since $j\leq \ell-1$, their entry $n_2$ is strictly smaller than the corresponding entry in (b.1). On the other hand, if the pair 
$(n/2,n/2)$ is added to $QB(n)$, it is not of the form (b.1) to (b.3), because all these pairs have both entries divisible by $2^k$, while  the maximum power of 2 that divides $n/2$ is $2^{k-1}$. Finally, if $(n_1,n_2)$ is a pair of the form (b.2), then $n_1/2^k$ is even and $n_2/2^k$ is odd, while   if $(n_1,n_2)$ is a pair of the form (b.3), then $n_1/2^k$ is odd and $n_2/2^k$ is even. Therefore, no pair can be simultaneously of the form (b.2) and (b.3).  
\end{proof}

\begin{exm}
Let us find $QB(214)$. Since $214=2(2^6+2^5+2^3+2+1)$, with the notations of the last corollary we have that $k=1$, $\ell=5$, $m_1=6$, $m_2=5$, $m_3=3$, $m_4=1$, and $m_5=0$. Then:
\begin{itemize}
\item[(b.1)]  The pair of this type  in $QB(214)$ is 
$\big(2^k(\sum_{i=1}^4 2^{m_i-1}+1),
2^k\sum_{i=1}^4 2^{m_i-1}\big)=(108,106)$.

\item[(b.2)] The indices $j\in \{2,3,4\}$ such that $m_j>m_{j+1}+1$ are 2 and 3. Therefore, the pairs of this type  in $QB(214)$ are:
\begin{itemize}
\item For $j=2$, $\big(2^{k}(2^{m_1-1}+2^{m_2}),n-2^{k}(2^{m_1-1}+2^{m_2})\big)=(128,86)$.

\item For $j=3$, $\big(2^{k}(2^{m_1-1}+2^{m_2-1}+2^{m_3}),n-2^{k}(2^{m_1-1}+2^{m_2-1}+2^{m_3})\big)=(112,102)$.

\end{itemize}

\item[(b.3)] The indices $j\in \{2,3,4\}$ such $m_j<m_{j-1}-1$ are 3 and 4. Therefore, the pairs of this type  in $QB(214)$ are:
\begin{itemize}
\item For $j=3$, $\big(n-2^k(2^{m_1-1}+2^{m_2-1}),2^k(2^{m_1-1}+2^{m_2-1})\big)=(118,96)$.

\item For $j=4$, $\big(n-2^k(2^{m_1-1}+2^{m_2-1}+2^{m_3-1}),2^k(2^{m_1-1}+2^{m_2-1}+2^{m_3-1})\big)=(110,104)$.
\end{itemize}

\item[(b.4)]  Since $214=2\cdot 107$ is even, $QB(214)$ contains the pair $(107,107)$.

\end{itemize}
Therefore
$$
QB(214)=\big\{(107,107),(108,106),(110,104),(112,102),(118,96),(128,86)\big\}.
$$
\end{exm}

\begin{cor}
For every $n\geq 2$, the cardinality of $QB(n)$ is at most $ \floor*{\log_2(n)}$.
\end{cor}

\begin{proof}
Let $n_{(2)}$ denote the binary representation of $n$. If $n$ is a power of 2, then $|QB(n)|=1\leq \floor*{\log_2(n)}$. Assume henceforth that $n$ is not a power of 2. In this case, by construction, the number of pairs of type (b.2) in $QB(n)$ is the number of maximal sequences of zeroes in $n_{(2)}$ that do not end immediately before the last 1 or  in the units position;  the number of pairs of type (b.3) in $QB(n)$ is the number of maximal sequences of zeroes in $n_{(2)}$ that do not start immediately after the leading 1 or that do not end in the units position;  there is one pair of type (b.4) in $QB(n)$ if $n_{(2)}$ contains a sequence of zeroes ending in the units position; and $QB(n)$ always contains a pair of the form (b.1). 
So, if we denote by $M_0(n)$ the number of maximal sequences of zeroes in $n_{(2)}$,  to compute the cardinality $|QB(n)|$:
\begin{itemize}
\item We count twice the number of  maximal sequences of zeroes in $n_{(2)}$ plus 1, $2M_0(n)+1$
\item We subtract 1 if $n_{(2)}$ contains a maximal sequence of zeroes  starting immediately after the leading 1
\item We subtract 1 if $n_{(2)}$ contains a maximal sequence of zeroes  ending immediately before the last 1
\item We subtract 2 and we add 1 (i.e., we subtract 1)  if $n_{(2)}$ contains a maximal sequence of zeroes  ending in the units position
\end{itemize}
So, if, to simplify the language, we call \emph{forbidden} any maximal sequence of zeroes in $n_{(2)}$ that starts immediately after the leading 1 or ends immediately before the last 1 or in the units position, we have the formula
$$
\begin{array}{rl}
|QB(n)|=2M_0(n)+1 &\mbox{minus the number of forbidden maximal sequences}\\ &\mbox{of zeroes in $n_{(2)}$,}
\end{array}
$$
where in the subtraction we count each forbidden maximal sequence of zeroes as many times as it satisfies a ``forbidden'' property. So, a maximal sequence of zeroes starting immediately after the leading 1 and  ending immediately before the last 1 subtracts 2. 

Now, on the one hand, if $\floor*{\log_2(n)}$ is an even number, by the pigeonhole principle we have that $M_0(n)\leq \floor*{\log_2(n)}/2$. But if $n_{(2)}$ does not contain any forbidden maximal sequence of zeroes, then $n_{(2)}$ starts with $11$ and ends with $11$ and the number of maximal sequences of zeroes in such an $n_{(2)}$ is at most $\floor*{\log_2(n)}/2-1$. So, if $M_0(n)=\floor*{\log_2(n)}/2$, then $n_{(2)}$ contains some forbidden maximal sequence of zeroes and then
$|QB(n)|\leq 2M_0(n)=\floor*{\log_2(n)}$, while if $M_0(n)\leq \floor*{\log_2(n)}/2-1$, then $|QB(n)|\leq 2M_0(n)+1\leq \floor*{\log_2(n)}-1$.

On the other hand,  if $\floor*{\log_2(n)}$ is an odd number, again by the pigeonhole principle we have that $M_0(n)\leq (\floor*{\log_2(n)}+1)/2$. Now, if $M_0(n)=(\floor*{\log_2(n)}+1)/2$, then $n_{(2)}$ contains at least 2 forbidden maximal sequences of zeroes. Indeed, let $\floor*{\log_2(n)}=2s+1$. If $n_{(2)}$ starts with $11$, avoiding a forbidden maximal sequence of zeroes at the beginning, then $M_0(n)\leq s=(\floor*{\log_2(n)}-1)/2$.  On the other hand, if it ends in $11$, avoiding a forbidden maximal sequence of zeroes at the end, then again $M_0(n)\leq s=(\floor*{\log_2(n)}-1)/2$.
So, to reach the maximum value of $M_0(n)$, $n_{(2)}$ must start with $10$ and end with $10$, $01$ or $00$, thus having at least 2 forbidden maximal sequences of zeroes. 
Thus, if $M_0(n)=(\floor*{\log_2(n)}+1)/2$, then
$|QB(n)|\leq 2M_0(n)-1=\floor*{\log_2(n)}$, while if $M_0(n)\leq (\floor*{\log_2(n)}+1)/2-1$, then $|QB(n)|\leq 2M_0(n)+1\leq \floor*{\log_2(n)}$. 
\end{proof}

Proposition \ref{lem:charmin1}, together with Lemma \ref{lem:C=0}, provide the following algorithm to produce all trees $T\in \TT_n$ such that $C(T)=C(n)$, which is reminiscent of Aldous' $\beta$-model~\citep{Ald1}.
\medskip

\noindent\textbf{Algorithm MinColless.}

\begin{enumerate}
\item[1)] Start with a single node labeled $n$.
\item[2)] While  the current tree contains labeled leaves, choose a leaf with label $m$.
\begin{itemize}
\item[2.a)] If $m$ is a power of 2, replace this leaf by a fully bifurcating tree $B_m$ with its nodes unlabeled.
\item[2.b)] If $m$ is a not a power of 2:
\begin{itemize}
\item[2.b.i)] Find a pair of integers $(m_1,m_2)\in QB(m)$.
\item[2.b.ii)] Split the leaf labeled $m$ into a cherry with unlabeled root and its leaves labeled $m_1$ and $m_2$, respectively.
\end{itemize}
\end{itemize}
\end{enumerate}

\begin{exm}\label{ex:2}
Let us use Algorithm MinColless to find all bifurcating trees with 20 leaves and minimum Colless index; we describe the trees by means of the usual Newick format,\footnote{See \url{http://evolution.genetics.washington.edu/phylip/newicktree.html}}   with the unlabeled leaves represented by a symbol $*$ and omitting the semicolon ending mark in order not to confuse it with a punctuation mark. 
\begin{itemize}
\item[1)] We start with a single node labeled 20. 
\item[2)] Since $QB(20)=\{(10,10),(12,8)\}$, this node splits into the cherries $(10,10)$ and $(12,8)$.
\item[3.1)] Since $QB(10)=\{(5,5),(6,4)\}$, the different ways of splitting the leaves of the tree $(10,10)$ produce the trees $((5,5),(5,5))$, $((5,5),(6,4))$, and $((6,4),(6,4))$.
Now, since $QB(5)=\{(3,2)\}$, $QB(6)=\{(3,3),(4,2)\}$, and $QB(3)=\{(2,1)\}$, and 1, 2,  and 4 are powers of 2,
we have the following derivations from these trees through all possible combinations of splitting the leaves in the trees:
$$
\begin{array}{l}
((5,5),(5,5))  \Rightarrow (((3,2),(3,2)),((3,2),(3,2)))  \Rightarrow ((((2,1),2),((2,1),2)),(((2,1),2),((2,1),2)))\\ \qquad \Rightarrow 
(((((*,*),*),(*,*)),(((*,*),*),(*,*))),((((*,*),*),(*,*)),(((*,*),*),(*,*))))\\
((5,5),(6,4)) \Rightarrow (((3,2),(3,2)),((3,3),4))  \Rightarrow  ((((2,1),2),((2,1),2)),(((2,1),(2,1)),4))\\ \qquad\Rightarrow  (((((*,*),*),(*,*)),(((*,*),*),(*,*))),((((*,*),*),((*,*),*)),((*,*),(*,*)))\\
((5,5),(6,4)) \Rightarrow (((3,2),(3,2)),((4,2),4))  \Rightarrow  ((((2,1),2),((2,1),2)),((4,2),4))\\ \qquad\Rightarrow  (((((*,*),*),(*,*)),(((*,*),*),(*,*))),((((*,*),(*,*)),(*,*)),((*,*),(*,*))))\\
((6,4),(6,4)) \Rightarrow (((3,3),4),((3,3),4))  \Rightarrow  ((((2,1),(2,1)),4),(((2,1),(2,1)),4))\\ \qquad \Rightarrow  (((((*,*),*),((*,*),*)),((*,*),(*,*))),((((*,*),*),((*,*),*)),((*,*),(*,*))))\\
((6,4),(6,4)) \Rightarrow (((3,3),4),((4,2),4))  \Rightarrow  ((((2,1),(2,1)),4),((4,2),4))\\ \qquad  \Rightarrow  (((((*,*),*),((*,*),*)),((*,*),(*,*))),((((*,*),(*,*)),(*,*)),((*,*),(*,*))))\\
((6,4),(6,4))  \Rightarrow (((4,2),4),((4,2),4)) \\ \qquad \Rightarrow  (((((*,*),(*,*)),(*,*)),((*,*),(*,*))),((((*,*),(*,*)),(*,*)),((*,*),(*,*))))
\end{array}
$$
\item[3.2)] Since $QB(12)=\{(6,6),(8,4)\}$ and 8 is a power of 2, the tree $(12,8)$ gives rise to the trees
$((6,6),8)$ and  $((8,4),8)$, and then, using $QB(6)=\{(3,3),(4,2)\}$ and $QB(3)=\{(2,1)\}$,
$$
\begin{array}{l}
((6,6),8)\Rightarrow (((3,3),(3,3)),8) \Rightarrow ((((2,1),(2,1)),((2,1),(2,1))),8) \\ \qquad \Rightarrow (((((*,*),*),((*,*),*)),(((*,*),*),((*,*),*))),(((*,*),(*,*)),((*,*),(*,*))))\\
((6,6),8)  \Rightarrow (((3,3),(4,2)),8)\Rightarrow ((((2,1),(2,1)),(4,2)),8) \\ \qquad \Rightarrow (((((*,*),*),((*,*),*)),(((*,*),(*,*)),(*,*))),(((*,*),(*,*)),((*,*),(*,*))))\\
((6,6),8)  \Rightarrow (((4,2),(4,2)),8)  \\ \qquad\Rightarrow (((((*,*),(*,*)),(*,*)),(((*,*),(*,*)),(*,*))),(((*,*),(*,*)),((*,*),(*,*))))\\
((8,4),8) \Rightarrow (((((*,*),(*,*)),((*,*),(*,*))),((*,*),(*,*))),(((*,*),(*,*)),((*,*),(*,*))))
\end{array}
$$
\end{itemize}
So, there are 10 different trees in $\TT_{20}$ with minimum Colles index.
\end{exm}

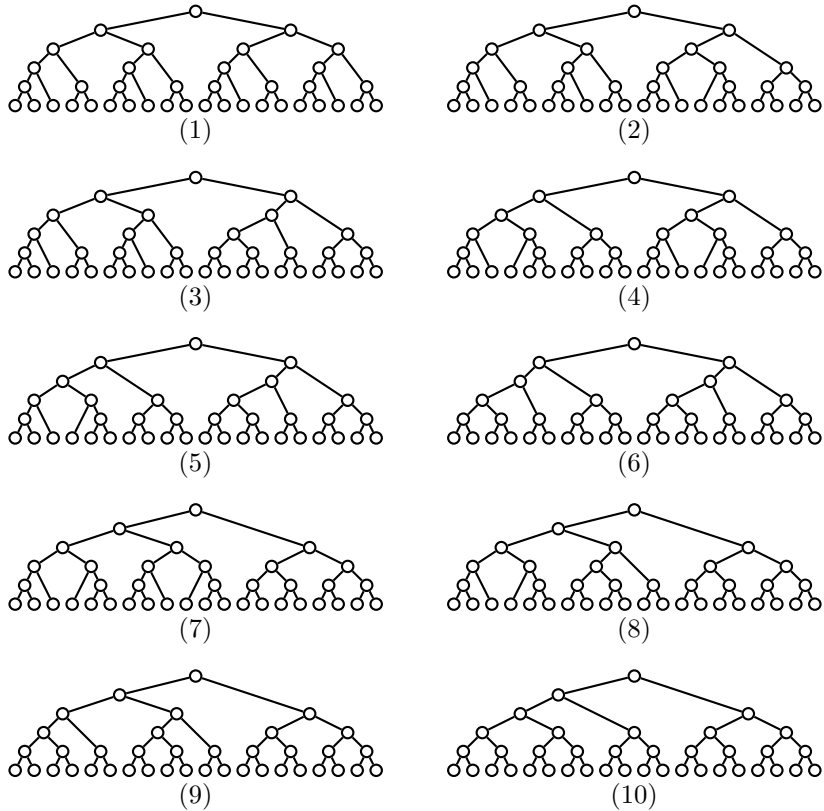
\begin{figure}[htb]
\begin{center}
\begin{tabular}{ccc}
\begin{tikzpicture}[thick,>=stealth,scale=0.25]
\draw(1,0) node[trepp] (l1) {}; 
\draw(2,0) node[trepp] (l2) {}; 
\draw(3,0) node[trepp] (l3) {}; 
\draw(4,0) node[trepp] (l4) {}; 
\draw(5,0) node[trepp] (l5) {}; 
\draw(6,0) node[trepp] (l6) {}; 
\draw(7,0) node[trepp] (l7) {}; 
\draw(8,0) node[trepp] (l8) {}; 
\draw(9,0) node[trepp] (l9) {}; 
\draw(10,0) node[trepp] (l10) {}; 
\draw(11,0) node[trepp] (l11) {}; 
\draw(12,0) node[trepp] (l12) {}; 
\draw(13,0) node[trepp] (l13) {}; 
\draw(14,0) node[trepp] (l14) {}; 
\draw(15,0) node[trepp] (l15) {}; 
\draw(16,0) node[trepp] (l16) {}; 
\draw(17,0) node[trepp] (l17) {}; 
\draw(18,0) node[trepp] (l18) {}; 
\draw(19,0) node[trepp] (l19) {}; 
\draw(20,0) node[trepp] (l20) {}; 
\draw(1.5,1) node[trepp] (v1) {}; 
\draw(4.5,1) node[trepp] (v2) {}; 
\draw(6.5,1) node[trepp] (v3) {}; 
\draw(9.5,1) node[trepp] (v4) {}; 
\draw(11.5,1) node[trepp] (v5) {}; 
\draw(14.5,1) node[trepp] (v6) {}; 
\draw(16.5,1) node[trepp] (v7) {}; 
\draw(19.5,1) node[trepp] (v8) {}; 
\draw(2,2) node[trepp] (w1) {}; 
\draw(7,2) node[trepp] (w2) {}; 
\draw(12,2) node[trepp] (w3) {}; 
\draw(17,2) node[trepp] (w4) {}; 
\draw(3,3) node[trepp] (x1) {}; 
\draw(8,3) node[trepp] (x2) {}; 
\draw(13,3) node[trepp] (x3) {}; 
\draw(18,3) node[trepp] (x4) {}; 
\draw(5.5,4) node[trepp] (y1) {}; 
\draw(15.5,4) node[trepp] (y2) {}; 
\draw(10.5,5) node[trepp] (r) {}; 
\draw (r)--(y1);
\draw (r)--(y2);
\draw (y1)--(x1);
\draw (y1)--(x2);
\draw (y2)--(x3);
\draw (y2)--(x4);
\draw (x1)--(w1);
\draw (x2)--(w2);
\draw (x3)--(w3);
\draw (x4)--(w4);
\draw (x1)--(v2);
\draw (x2)--(v4);
\draw (x3)--(v6);
\draw (x4)--(v8);
\draw (w1)--(v1);
\draw (w1)--(l3);
\draw (w2)--(v3);
\draw (w2)--(l8);
\draw (w3)--(v5);
\draw (w3)--(l13);
\draw (w4)--(v7);
\draw (w4)--(l18);
\draw (v1)--(l1);
\draw (v1)--(l2);
\draw (v2)--(l4);
\draw (v2)--(l5);
\draw (v3)--(l6);
\draw (v3)--(l7);
\draw (v4)--(l9);
\draw (v4)--(l10);
\draw (v5)--(l11);
\draw (v5)--(l12);
\draw (v6)--(l14);
\draw (v6)--(l15);
\draw (v7)--(l16);
\draw (v7)--(l17);
\draw (v8)--(l19);
\draw (v8)--(l20);
\end{tikzpicture}
&\quad&
\begin{tikzpicture}[thick,>=stealth,scale=0.25]
\draw(1,0) node[trepp] (l1) {}; 
\draw(2,0) node[trepp] (l2) {}; 
\draw(3,0) node[trepp] (l3) {}; 
\draw(4,0) node[trepp] (l4) {}; 
\draw(5,0) node[trepp] (l5) {}; 
\draw(6,0) node[trepp] (l6) {}; 
\draw(7,0) node[trepp] (l7) {}; 
\draw(8,0) node[trepp] (l8) {}; 
\draw(9,0) node[trepp] (l9) {}; 
\draw(10,0) node[trepp] (l10) {}; 
\draw(11,0) node[trepp] (l11) {}; 
\draw(12,0) node[trepp] (l12) {}; 
\draw(13,0) node[trepp] (l13) {}; 
\draw(14,0) node[trepp] (l14) {}; 
\draw(15,0) node[trepp] (l15) {}; 
\draw(16,0) node[trepp] (l16) {}; 
\draw(17,0) node[trepp] (l17) {}; 
\draw(18,0) node[trepp] (l18) {}; 
\draw(19,0) node[trepp] (l19) {}; 
\draw(20,0) node[trepp] (l20) {}; 
\draw(1.5,1) node[trepp] (v1) {}; 
\draw(4.5,1) node[trepp] (v2) {}; 
\draw(6.5,1) node[trepp] (v3) {}; 
\draw(9.5,1) node[trepp] (v4) {}; 
\draw(11.5,1) node[trepp] (v5) {}; 
\draw(15.5,1) node[trepp] (v6) {}; 
\draw(17.5,1) node[trepp] (v7) {}; 
\draw(19.5,1) node[trepp] (v8) {}; 
\draw(2,2) node[trepp] (w1) {}; 
\draw(7,2) node[trepp] (w2) {}; 
\draw(12,2) node[trepp] (w3) {}; 
\draw(15,2) node[trepp] (w4) {}; 
\draw(3,3) node[trepp] (x1) {}; 
\draw(8,3) node[trepp] (x2) {}; 
\draw(13.5,3) node[trepp] (x3) {}; 
\draw(18.5,2) node[trepp] (x4) {}; 
\draw(5.5,4) node[trepp] (y1) {}; 
\draw(15.5,4) node[trepp] (y2) {}; 
\draw(10.5,5) node[trepp] (r) {}; 
\draw (r)--(y1);
\draw (r)--(y2);
\draw (y1)--(x1);
\draw (y1)--(x2);
\draw (y2)--(x3);
\draw (y2)--(x4);
\draw (x1)--(w1);
\draw (x2)--(w2);
\draw (x3)--(w3);
\draw (x3)--(w4);
\draw (x1)--(v2);
\draw (x2)--(v4);
\draw (x4)--(v7);
\draw (x4)--(v8);
\draw (w1)--(v1);
\draw (w1)--(l3);
\draw (w2)--(v3);
\draw (w2)--(l8);
\draw (w3)--(v5);
\draw (w3)--(l13);
\draw (w4)--(v6);
\draw (w4)--(l14);
\draw (v1)--(l1);
\draw (v1)--(l2);
\draw (v2)--(l4);
\draw (v2)--(l5);
\draw (v3)--(l6);
\draw (v3)--(l7);
\draw (v4)--(l9);
\draw (v4)--(l10);
\draw (v5)--(l11);
\draw (v5)--(l12);
\draw (v6)--(l15);
\draw (v6)--(l16);
\draw (v7)--(l18);
\draw (v7)--(l17);
\draw (v8)--(l19);
\draw (v8)--(l20);
\end{tikzpicture}
\\[-0.5ex]
 (1) & & (2)\\[2ex]
\begin{tikzpicture}[thick,>=stealth,scale=0.25]
\draw(1,0) node[trepp] (l1) {}; 
\draw(2,0) node[trepp] (l2) {}; 
\draw(3,0) node[trepp] (l3) {}; 
\draw(4,0) node[trepp] (l4) {}; 
\draw(5,0) node[trepp] (l5) {}; 
\draw(6,0) node[trepp] (l6) {}; 
\draw(7,0) node[trepp] (l7) {}; 
\draw(8,0) node[trepp] (l8) {}; 
\draw(9,0) node[trepp] (l9) {}; 
\draw(10,0) node[trepp] (l10) {}; 
\draw(11,0) node[trepp] (l11) {}; 
\draw(12,0) node[trepp] (l12) {}; 
\draw(13,0) node[trepp] (l13) {}; 
\draw(14,0) node[trepp] (l14) {}; 
\draw(15,0) node[trepp] (l15) {}; 
\draw(16,0) node[trepp] (l16) {}; 
\draw(17,0) node[trepp] (l17) {}; 
\draw(18,0) node[trepp] (l18) {}; 
\draw(19,0) node[trepp] (l19) {}; 
\draw(20,0) node[trepp] (l20) {}; 
\draw(1.5,1) node[trepp] (v1) {}; 
\draw(4.5,1) node[trepp] (v2) {}; 
\draw(6.5,1) node[trepp] (v3) {}; 
\draw(9.5,1) node[trepp] (v4) {}; 
\draw(11.5,1) node[trepp] (v5) {}; 
\draw(13.5,1) node[trepp] (v6) {}; 
\draw(15.5,1) node[trepp] (v7) {}; 
\draw(17.5,1) node[trepp] (v8) {}; 
\draw(19.5,1) node[trepp] (v9) {}; 
\draw(2,2) node[trepp] (w1) {}; 
\draw(7,2) node[trepp] (w2) {}; 
\draw(12.5,2) node[trepp] (w3) {}; 
\draw(18.5,2) node[trepp] (w4) {}; 
\draw(3,3) node[trepp] (x1) {}; 
\draw(8,3) node[trepp] (x2) {}; 
\draw(14.5,3) node[trepp] (x3) {}; 
\draw(5.5,4) node[trepp] (y1) {}; 
\draw(15.5,4) node[trepp] (y2) {}; 
\draw(10.5,5) node[trepp] (r) {}; 
\draw (r)--(y1);
\draw (r)--(y2);
\draw (y1)--(x1);
\draw (y1)--(x2);
\draw (y2)--(x3);
\draw (y2)--(w4);
\draw (x1)--(w1);
\draw (x2)--(w2);
\draw (x3)--(w3);
\draw (x3)--(v7);
\draw (x1)--(v2);
\draw (x2)--(v4);
\draw (w1)--(v1);
\draw (w1)--(l3);
\draw (w2)--(v3);
\draw (w2)--(l8);
\draw (w3)--(v5);
\draw (w3)--(v6);
\draw (w4)--(v8);
\draw (w4)--(v9);
\draw (v1)--(l1);
\draw (v1)--(l2);
\draw (v2)--(l4);
\draw (v2)--(l5);
\draw (v3)--(l6);
\draw (v3)--(l7);
\draw (v4)--(l9);
\draw (v4)--(l10);
\draw (v5)--(l11);
\draw (v5)--(l12);
\draw (v6)--(l13);
\draw (v6)--(l14);
\draw (v7)--(l15);
\draw (v7)--(l16);
\draw (v8)--(l17);
\draw (v8)--(l18);
\draw (v9)--(l19);
\draw (v9)--(l20);
\end{tikzpicture}
&\quad&
\begin{tikzpicture}[thick,>=stealth,scale=0.25]
\draw(1,0) node[trepp] (l1) {}; 
\draw(2,0) node[trepp] (l2) {}; 
\draw(3,0) node[trepp] (l3) {}; 
\draw(4,0) node[trepp] (l4) {}; 
\draw(5,0) node[trepp] (l5) {}; 
\draw(6,0) node[trepp] (l6) {}; 
\draw(7,0) node[trepp] (l7) {}; 
\draw(8,0) node[trepp] (l8) {}; 
\draw(9,0) node[trepp] (l9) {}; 
\draw(10,0) node[trepp] (l10) {}; 
\draw(11,0) node[trepp] (l11) {}; 
\draw(12,0) node[trepp] (l12) {}; 
\draw(13,0) node[trepp] (l13) {}; 
\draw(14,0) node[trepp] (l14) {}; 
\draw(15,0) node[trepp] (l15) {}; 
\draw(16,0) node[trepp] (l16) {}; 
\draw(17,0) node[trepp] (l17) {}; 
\draw(18,0) node[trepp] (l18) {}; 
\draw(19,0) node[trepp] (l19) {}; 
\draw(20,0) node[trepp] (l20) {}; 
\draw(1.5,1) node[trepp] (v1) {}; 
\draw(5.5,1) node[trepp] (v2) {}; 
\draw(7.5,1) node[trepp] (v3) {}; 
\draw(9.5,1) node[trepp] (v4) {}; 
\draw(11.5,1) node[trepp] (v5) {}; 
\draw(15.5,1) node[trepp] (v6) {}; 
\draw(17.5,1) node[trepp] (v7) {}; 
\draw(19.5,1) node[trepp] (v8) {}; 
\draw(2,2) node[trepp] (w1) {}; 
\draw(5,2) node[trepp] (w2) {}; 
\draw(12,2) node[trepp] (w3) {}; 
\draw(15,2) node[trepp] (w4) {}; 
\draw(3.5,3) node[trepp] (x1) {}; 
\draw(8.5,2) node[trepp] (x2) {}; 
\draw(13.5,3) node[trepp] (x3) {}; 
\draw(18.5,2) node[trepp] (x4) {}; 
\draw(5.5,4) node[trepp] (y1) {}; 
\draw(15.5,4) node[trepp] (y2) {}; 
\draw(10.5,5) node[trepp] (r) {}; 
\draw (r)--(y1);
\draw (r)--(y2);
\draw (y1)--(x1);
\draw (y1)--(x2);
\draw (y2)--(x3);
\draw (y2)--(x4);
\draw (x1)--(w1);
\draw (x1)--(w2);
\draw (x2)--(v3);
\draw (x2)--(v4);
\draw (x3)--(w3);
\draw (x3)--(w4);
\draw (x4)--(v7);
\draw (x4)--(v8);
\draw (w1)--(v1);
\draw (w1)--(l3);
\draw (w2)--(v2);
\draw (w2)--(l4);
\draw (w3)--(v5);
\draw (w3)--(l13);
\draw (w4)--(v6);
\draw (w4)--(l14);
\draw (v1)--(l1);
\draw (v1)--(l2);
\draw (v2)--(l5);
\draw (v2)--(l6);
\draw (v3)--(l7);
\draw (v3)--(l8);
\draw (v4)--(l9);
\draw (v4)--(l10);
\draw (v5)--(l11);
\draw (v5)--(l12);
\draw (v6)--(l15);
\draw (v6)--(l16);
\draw (v7)--(l18);
\draw (v7)--(l17);
\draw (v8)--(l19);
\draw (v8)--(l20);
\end{tikzpicture}
\\[-0.5ex]
 (3) & & (4)\\[2ex]
\begin{tikzpicture}[thick,>=stealth,scale=0.25]
\draw(1,0) node[trepp] (l1) {}; 
\draw(2,0) node[trepp] (l2) {}; 
\draw(3,0) node[trepp] (l3) {}; 
\draw(4,0) node[trepp] (l4) {}; 
\draw(5,0) node[trepp] (l5) {}; 
\draw(6,0) node[trepp] (l6) {}; 
\draw(7,0) node[trepp] (l7) {}; 
\draw(8,0) node[trepp] (l8) {}; 
\draw(9,0) node[trepp] (l9) {}; 
\draw(10,0) node[trepp] (l10) {}; 
\draw(11,0) node[trepp] (l11) {}; 
\draw(12,0) node[trepp] (l12) {}; 
\draw(13,0) node[trepp] (l13) {}; 
\draw(14,0) node[trepp] (l14) {}; 
\draw(15,0) node[trepp] (l15) {}; 
\draw(16,0) node[trepp] (l16) {}; 
\draw(17,0) node[trepp] (l17) {}; 
\draw(18,0) node[trepp] (l18) {}; 
\draw(19,0) node[trepp] (l19) {}; 
\draw(20,0) node[trepp] (l20) {}; 
\draw(1.5,1) node[trepp] (v1) {}; 
\draw(5.5,1) node[trepp] (v2) {}; 
\draw(7.5,1) node[trepp] (v3) {}; 
\draw(9.5,1) node[trepp] (v4) {}; 
\draw(11.5,1) node[trepp] (v5) {}; 
\draw(13.5,1) node[trepp] (v6) {}; 
\draw(15.5,1) node[trepp] (v7) {}; 
\draw(17.5,1) node[trepp] (v8) {}; 
\draw(19.5,1) node[trepp] (v9) {}; 
\draw(2,2) node[trepp] (w1) {}; 
\draw(5,2) node[trepp] (w2) {}; 
\draw(12.5,2) node[trepp] (w3) {}; 
\draw(18.5,2) node[trepp] (w4) {}; 
\draw(3.5,3) node[trepp] (x1) {}; 
\draw(8.5,2) node[trepp] (x2) {}; 
\draw(14.5,3) node[trepp] (x3) {}; 
\draw(5.5,4) node[trepp] (y1) {}; 
\draw(15.5,4) node[trepp] (y2) {}; 
\draw(10.5,5) node[trepp] (r) {}; 
\draw (r)--(y1);
\draw (r)--(y2);
\draw (y1)--(x1);
\draw (y1)--(x2);
\draw (y2)--(x3);
\draw (y2)--(w4);
\draw (x1)--(w1);
\draw (x1)--(w2);
\draw (x2)--(v3);
\draw (x2)--(v4);
\draw (x3)--(w3);
\draw (x3)--(v7);
\draw (w1)--(v1);
\draw (w1)--(l3);
\draw (w2)--(v2);
\draw (w2)--(l4);
\draw (w3)--(v5);
\draw (w3)--(v6);
\draw (w4)--(v8);
\draw (w4)--(v9);
\draw (v1)--(l1);
\draw (v1)--(l2);
\draw (v2)--(l5);
\draw (v2)--(l6);
\draw (v3)--(l7);
\draw (v3)--(l8);
\draw (v4)--(l9);
\draw (v4)--(l10);
\draw (v5)--(l11);
\draw (v5)--(l12);
\draw (v6)--(l13);
\draw (v6)--(l14);
\draw (v7)--(l15);
\draw (v7)--(l16);
\draw (v8)--(l17);
\draw (v8)--(l18);
\draw (v9)--(l19);
\draw (v9)--(l20);
\end{tikzpicture}
&\quad&
\begin{tikzpicture}[thick,>=stealth,scale=0.25]
\draw(1,0) node[trepp] (l1) {}; 
\draw(2,0) node[trepp] (l2) {}; 
\draw(3,0) node[trepp] (l3) {}; 
\draw(4,0) node[trepp] (l4) {}; 
\draw(5,0) node[trepp] (l5) {}; 
\draw(6,0) node[trepp] (l6) {}; 
\draw(7,0) node[trepp] (l7) {}; 
\draw(8,0) node[trepp] (l8) {}; 
\draw(9,0) node[trepp] (l9) {}; 
\draw(10,0) node[trepp] (l10) {}; 
\draw(11,0) node[trepp] (l11) {}; 
\draw(12,0) node[trepp] (l12) {}; 
\draw(13,0) node[trepp] (l13) {}; 
\draw(14,0) node[trepp] (l14) {}; 
\draw(15,0) node[trepp] (l15) {}; 
\draw(16,0) node[trepp] (l16) {}; 
\draw(17,0) node[trepp] (l17) {}; 
\draw(18,0) node[trepp] (l18) {}; 
\draw(19,0) node[trepp] (l19) {}; 
\draw(20,0) node[trepp] (l20) {}; 
\draw(1.5,1) node[trepp] (v1) {}; 
\draw(3.5,1) node[trepp] (v2) {}; 
\draw(5.5,1) node[trepp] (v3) {}; 
\draw(7.5,1) node[trepp] (v4) {}; 
\draw(9.5,1) node[trepp] (v51) {}; 
\draw(11.5,1) node[trepp] (v5) {}; 
\draw(13.5,1) node[trepp] (v6) {}; 
\draw(15.5,1) node[trepp] (v7) {}; 
\draw(17.5,1) node[trepp] (v8) {}; 
\draw(19.5,1) node[trepp] (v9) {}; 
\draw(2.5,2) node[trepp] (w1) {}; 
\draw(8.5,2) node[trepp] (w2) {}; 
\draw(12.5,2) node[trepp] (w3) {}; 
\draw(18.5,2) node[trepp] (w4) {}; 
\draw(4.5,3) node[trepp] (x1) {}; 
\draw(14.5,3) node[trepp] (x3) {}; 
\draw(5.5,4) node[trepp] (y1) {}; 
\draw(15.5,4) node[trepp] (y2) {}; 
\draw(10.5,5) node[trepp] (r) {}; 
\draw (r)--(y1);
\draw (r)--(y2);
\draw (y1)--(x1);
\draw (y1)--(w2);
\draw (y2)--(x3);
\draw (y2)--(w4);
\draw (x1)--(w1);
\draw (x1)--(v3);
\draw (x3)--(w3);
\draw (x3)--(v7);
\draw (w1)--(v1);
\draw (w1)--(v2);
\draw (w2)--(v4);
\draw (w2)--(v51);
\draw (w3)--(v5);
\draw (w3)--(v6);
\draw (w4)--(v8);
\draw (w4)--(v9);
\draw (v1)--(l1);
\draw (v1)--(l2);
\draw (v2)--(l3);
\draw (v2)--(l4);
\draw (v3)--(l5);
\draw (v3)--(l6);
\draw (v4)--(l7);
\draw (v4)--(l8);
\draw (v51)--(l9);
\draw (v51)--(l10);
\draw (v5)--(l11);
\draw (v5)--(l12);
\draw (v6)--(l13);
\draw (v6)--(l14);
\draw (v7)--(l15);
\draw (v7)--(l16);
\draw (v8)--(l17);
\draw (v8)--(l18);
\draw (v9)--(l19);
\draw (v9)--(l20);
\end{tikzpicture}\\[-0.5ex]
 (5) & & (6)\\[2ex]
\begin{tikzpicture}[thick,>=stealth,scale=0.25]
\draw(1,0) node[trepp] (l1) {}; 
\draw(2,0) node[trepp] (l2) {}; 
\draw(3,0) node[trepp] (l3) {}; 
\draw(4,0) node[trepp] (l4) {}; 
\draw(5,0) node[trepp] (l5) {}; 
\draw(6,0) node[trepp] (l6) {}; 
\draw(7,0) node[trepp] (l7) {}; 
\draw(8,0) node[trepp] (l8) {}; 
\draw(9,0) node[trepp] (l9) {}; 
\draw(10,0) node[trepp] (l10) {}; 
\draw(11,0) node[trepp] (l11) {}; 
\draw(12,0) node[trepp] (l12) {}; 
\draw(13,0) node[trepp] (l13) {}; 
\draw(14,0) node[trepp] (l14) {}; 
\draw(15,0) node[trepp] (l15) {}; 
\draw(16,0) node[trepp] (l16) {}; 
\draw(17,0) node[trepp] (l17) {}; 
\draw(18,0) node[trepp] (l18) {}; 
\draw(19,0) node[trepp] (l19) {}; 
\draw(20,0) node[trepp] (l20) {}; 
\draw(1.5,1) node[trepp] (v1) {}; 
\draw(5.5,1) node[trepp] (v2) {}; 
\draw(7.5,1) node[trepp] (v3) {}; 
\draw(11.5,1) node[trepp] (v4) {}; 
\draw(13.5,1) node[trepp] (v6) {}; 
\draw(15.5,1) node[trepp] (v7) {}; 
\draw(17.5,1) node[trepp] (v8) {}; 
\draw(19.5,1) node[trepp] (v9) {}; 
\draw(2,2) node[trepp] (w1) {}; 
\draw(5,2) node[trepp] (w2) {}; 
\draw(8,2) node[trepp] (w11) {}; 
\draw(11,2) node[trepp] (w21) {}; 
\draw(14.5,2) node[trepp] (w3) {}; 
\draw(18.5,2) node[trepp] (w4) {}; 
\draw(3.5,3) node[trepp] (x1) {}; 
\draw(9.5,3) node[trepp] (x2) {}; 
\draw(6.5,4) node[trepp] (y1) {}; 
\draw(16.5,3) node[trepp] (y2) {}; 
\draw(10.5,5) node[trepp] (r) {}; 
\draw (r)--(y1);
\draw (r)--(y2);
\draw (y1)--(x1);
\draw (y1)--(x2);
\draw (y2)--(w3);
\draw (y2)--(w4);
\draw (x1)--(w1);
\draw (x1)--(w2);
\draw (x2)--(w11);
\draw (x2)--(w21);
\draw (w1)--(v1);
\draw (w1)--(l3);
\draw (w2)--(v2);
\draw (w2)--(l4);
\draw (w11)--(v3);
\draw (w11)--(l9);
\draw (w21)--(v4);
\draw (w21)--(l10);
\draw (w3)--(v6);
\draw (w3)--(v7);
\draw (w4)--(v8);
\draw (w4)--(v9);
\draw (v1)--(l1);
\draw (v1)--(l2);
\draw (v2)--(l5);
\draw (v2)--(l6);
\draw (v3)--(l7);
\draw (v3)--(l8);
\draw (v5)--(l11);
\draw (v5)--(l12);
\draw (v6)--(l13);
\draw (v6)--(l14);
\draw (v7)--(l15);
\draw (v7)--(l16);
\draw (v8)--(l17);
\draw (v8)--(l18);
\draw (v9)--(l19);
\draw (v9)--(l20);
\end{tikzpicture}
&\quad&
\begin{tikzpicture}[thick,>=stealth,scale=0.25]
\draw(1,0) node[trepp] (l1) {}; 
\draw(2,0) node[trepp] (l2) {}; 
\draw(3,0) node[trepp] (l3) {}; 
\draw(4,0) node[trepp] (l4) {}; 
\draw(5,0) node[trepp] (l5) {}; 
\draw(6,0) node[trepp] (l6) {}; 
\draw(7,0) node[trepp] (l7) {}; 
\draw(8,0) node[trepp] (l8) {}; 
\draw(9,0) node[trepp] (l9) {}; 
\draw(10,0) node[trepp] (l10) {}; 
\draw(11,0) node[trepp] (l11) {}; 
\draw(12,0) node[trepp] (l12) {}; 
\draw(13,0) node[trepp] (l13) {}; 
\draw(14,0) node[trepp] (l14) {}; 
\draw(15,0) node[trepp] (l15) {}; 
\draw(16,0) node[trepp] (l16) {}; 
\draw(17,0) node[trepp] (l17) {}; 
\draw(18,0) node[trepp] (l18) {}; 
\draw(19,0) node[trepp] (l19) {}; 
\draw(20,0) node[trepp] (l20) {}; 
\draw(1.5,1) node[trepp] (v1) {}; 
\draw(5.5,1) node[trepp] (v2) {}; 
\draw(7.5,1) node[trepp] (v3) {}; 
\draw(9.5,1) node[trepp] (v4) {}; 
\draw(11.5,1) node[trepp] (v5) {}; 
\draw(13.5,1) node[trepp] (v6) {}; 
\draw(15.5,1) node[trepp] (v7) {}; 
\draw(17.5,1) node[trepp] (v8) {}; 
\draw(19.5,1) node[trepp] (v9) {}; 
\draw(2,2) node[trepp] (w1) {}; 
\draw(5,2) node[trepp] (w2) {}; 
\draw(8.5,2) node[trepp] (w11) {}; 
\draw(14.5,2) node[trepp] (w3) {}; 
\draw(18.5,2) node[trepp] (w4) {}; 
\draw(3.5,3) node[trepp] (x1) {}; 
\draw(9.5,3) node[trepp] (x2) {}; 
\draw(6.5,4) node[trepp] (y1) {}; 
\draw(16.5,3) node[trepp] (y2) {}; 
\draw(10.5,5) node[trepp] (r) {}; 
\draw (r)--(y1);
\draw (r)--(y2);
\draw (y1)--(x1);
\draw (y1)--(x2);
\draw (y2)--(w3);
\draw (y2)--(w4);
\draw (x1)--(w1);
\draw (x1)--(w2);
\draw (x2)--(w11);
\draw (x2)--(v5);
\draw (w1)--(v1);
\draw (w1)--(l3);
\draw (w2)--(v2);
\draw (w2)--(l4);
\draw (w11)--(v3);
\draw (w11)--(v4);
\draw (w3)--(v6);
\draw (w3)--(v7);
\draw (w4)--(v8);
\draw (w4)--(v9);
\draw (v1)--(l1);
\draw (v1)--(l2);
\draw (v2)--(l5);
\draw (v2)--(l6);
\draw (v3)--(l7);
\draw (v3)--(l8);
\draw (v4)--(l9);
\draw (v4)--(l10);
\draw (v5)--(l11);
\draw (v5)--(l12);
\draw (v6)--(l13);
\draw (v6)--(l14);
\draw (v7)--(l15);
\draw (v7)--(l16);
\draw (v8)--(l17);
\draw (v8)--(l18);
\draw (v9)--(l19);
\draw (v9)--(l20);
\end{tikzpicture}\\[-0.5ex]
 (7) & & (8)\\[2ex]
\begin{tikzpicture}[thick,>=stealth,scale=0.25]
\draw(1,0) node[trepp] (l1) {}; 
\draw(2,0) node[trepp] (l2) {}; 
\draw(3,0) node[trepp] (l3) {}; 
\draw(4,0) node[trepp] (l4) {}; 
\draw(5,0) node[trepp] (l5) {}; 
\draw(6,0) node[trepp] (l6) {}; 
\draw(7,0) node[trepp] (l7) {}; 
\draw(8,0) node[trepp] (l8) {}; 
\draw(9,0) node[trepp] (l9) {}; 
\draw(10,0) node[trepp] (l10) {}; 
\draw(11,0) node[trepp] (l11) {}; 
\draw(12,0) node[trepp] (l12) {}; 
\draw(13,0) node[trepp] (l13) {}; 
\draw(14,0) node[trepp] (l14) {}; 
\draw(15,0) node[trepp] (l15) {}; 
\draw(16,0) node[trepp] (l16) {}; 
\draw(17,0) node[trepp] (l17) {}; 
\draw(18,0) node[trepp] (l18) {}; 
\draw(19,0) node[trepp] (l19) {}; 
\draw(20,0) node[trepp] (l20) {}; 
\draw(1.5,1) node[trepp] (v1) {}; 
\draw(3.5,1) node[trepp] (v2) {}; 
\draw(5.5,1) node[trepp] (v3) {}; 
\draw(7.5,1) node[trepp] (v4) {}; 
\draw(9.5,1) node[trepp] (v45) {}; 
\draw(11.5,1) node[trepp] (v5) {}; 
\draw(13.5,1) node[trepp] (v6) {}; 
\draw(15.5,1) node[trepp] (v7) {}; 
\draw(17.5,1) node[trepp] (v8) {}; 
\draw(19.5,1) node[trepp] (v9) {}; 
\draw(2.5,2) node[trepp] (w1) {}; 
\draw(8.5,2) node[trepp] (w11) {}; 
\draw(14.5,2) node[trepp] (w3) {}; 
\draw(18.5,2) node[trepp] (w4) {}; 
\draw(3.5,3) node[trepp] (x1) {}; 
\draw(9.5,3) node[trepp] (x2) {}; 
\draw(6.5,4) node[trepp] (y1) {}; 
\draw(16.5,3) node[trepp] (y2) {}; 
\draw(10.5,5) node[trepp] (r) {}; 
\draw (r)--(y1);
\draw (r)--(y2);
\draw (y1)--(x1);
\draw (y1)--(x2);
\draw (y2)--(w3);
\draw (y2)--(w4);
\draw (x1)--(w1);
\draw (x1)--(v3);
\draw (x2)--(w11);
\draw (x2)--(v5);
\draw (w1)--(v1);
\draw (w1)--(v2);
\draw (w11)--(v4);
\draw (w11)--(v45);
\draw (w3)--(v6);
\draw (w3)--(v7);
\draw (w4)--(v8);
\draw (w4)--(v9);
\draw (v1)--(l1);
\draw (v1)--(l2);
\draw (v2)--(l3);
\draw (v2)--(l4);
\draw (v3)--(l5);
\draw (v3)--(l6);
\draw (v4)--(l7);
\draw (v4)--(l8);
\draw (v45)--(l9);
\draw (v45)--(l10);
\draw (v5)--(l11);
\draw (v5)--(l12);
\draw (v6)--(l13);
\draw (v6)--(l14);
\draw (v7)--(l15);
\draw (v7)--(l16);
\draw (v8)--(l17);
\draw (v8)--(l18);
\draw (v9)--(l19);
\draw (v9)--(l20);
\end{tikzpicture}
&\quad&
\begin{tikzpicture}[thick,>=stealth,scale=0.25]
\draw(1,0) node[trepp] (l1) {}; 
\draw(2,0) node[trepp] (l2) {}; 
\draw(3,0) node[trepp] (l3) {}; 
\draw(4,0) node[trepp] (l4) {}; 
\draw(5,0) node[trepp] (l5) {}; 
\draw(6,0) node[trepp] (l6) {}; 
\draw(7,0) node[trepp] (l7) {}; 
\draw(8,0) node[trepp] (l8) {}; 
\draw(9,0) node[trepp] (l9) {}; 
\draw(10,0) node[trepp] (l10) {}; 
\draw(11,0) node[trepp] (l11) {}; 
\draw(12,0) node[trepp] (l12) {}; 
\draw(13,0) node[trepp] (l13) {}; 
\draw(14,0) node[trepp] (l14) {}; 
\draw(15,0) node[trepp] (l15) {}; 
\draw(16,0) node[trepp] (l16) {}; 
\draw(17,0) node[trepp] (l17) {}; 
\draw(18,0) node[trepp] (l18) {}; 
\draw(19,0) node[trepp] (l19) {}; 
\draw(20,0) node[trepp] (l20) {}; 
\draw(1.5,1) node[trepp] (v1) {}; 
\draw(3.5,1) node[trepp] (v2) {}; 
\draw(5.5,1) node[trepp] (v3) {}; 
\draw(7.5,1) node[trepp] (v4) {}; 
\draw(9.5,1) node[trepp] (v45) {}; 
\draw(11.5,1) node[trepp] (v5) {}; 
\draw(13.5,1) node[trepp] (v6) {}; 
\draw(15.5,1) node[trepp] (v7) {}; 
\draw(17.5,1) node[trepp] (v8) {}; 
\draw(19.5,1) node[trepp] (v9) {}; 
\draw(2.5,2) node[trepp] (w1) {}; 
\draw(6.5,2) node[trepp] (w2) {}; 
\draw(10.5,2) node[trepp] (w11) {}; 
\draw(14.5,2) node[trepp] (w3) {}; 
\draw(18.5,2) node[trepp] (w4) {}; 
\draw(4.5,3) node[trepp] (x1) {}; 
\draw(6.5,4) node[trepp] (y1) {}; 
\draw(16.5,3) node[trepp] (y2) {}; 
\draw(10.5,5) node[trepp] (r) {}; 
\draw (r)--(y1);
\draw (r)--(y2);
\draw (y1)--(x1);
\draw (y1)--(w11);
\draw (y2)--(w3);
\draw (y2)--(w4);
\draw (x1)--(w1);
\draw (x1)--(w2);
\draw (w1)--(v1);
\draw (w1)--(v2);
\draw (w2)--(v3);
\draw (w2)--(v4);
\draw (w11)--(v45);
\draw (w11)--(v5);
\draw (w3)--(v6);
\draw (w3)--(v7);
\draw (w4)--(v8);
\draw (w4)--(v9);
\draw (v1)--(l1);
\draw (v1)--(l2);
\draw (v2)--(l3);
\draw (v2)--(l4);
\draw (v3)--(l5);
\draw (v3)--(l6);
\draw (v4)--(l7);
\draw (v4)--(l8);
\draw (v45)--(l9);
\draw (v45)--(l10);
\draw (v5)--(l11);
\draw (v5)--(l12);
\draw (v6)--(l13);
\draw (v6)--(l14);
\draw (v7)--(l15);
\draw (v7)--(l16);
\draw (v8)--(l17);
\draw (v8)--(l18);
\draw (v9)--(l19);
\draw (v9)--(l20);
\end{tikzpicture}\\[-0.5ex]
 (9) & & (10)
 \end{tabular}
\end{center}
\caption{\label{fig:1020} The 10 trees in $\TT_{20}$ with minimum Colles index, 8. They are enumerated in the same order as they have been produced in Example \ref{ex:2}.}
\end{figure}

We have implemented the Algorithm MinColless, with step (2.b.i) efficiently carried out by means of Proposition \ref{cor:QB}, in a Python script that generates, for every $n$, the Newick description of all bifurcating trees in $\TT_n$ with minimum Colles index. It is available at the GitHub repository \url{https://github.com/biocom-uib/Colless}. As a proof of concept, we have computed for every $n$ from 4 to $2^6=128$ all such bifurcating trees in $\TT_n$. Figure \ref{fig:1} shows their number for every $n$. 

\begin{figure}[htb]
\begin{center}
\includegraphics[width=0.9\linewidth]{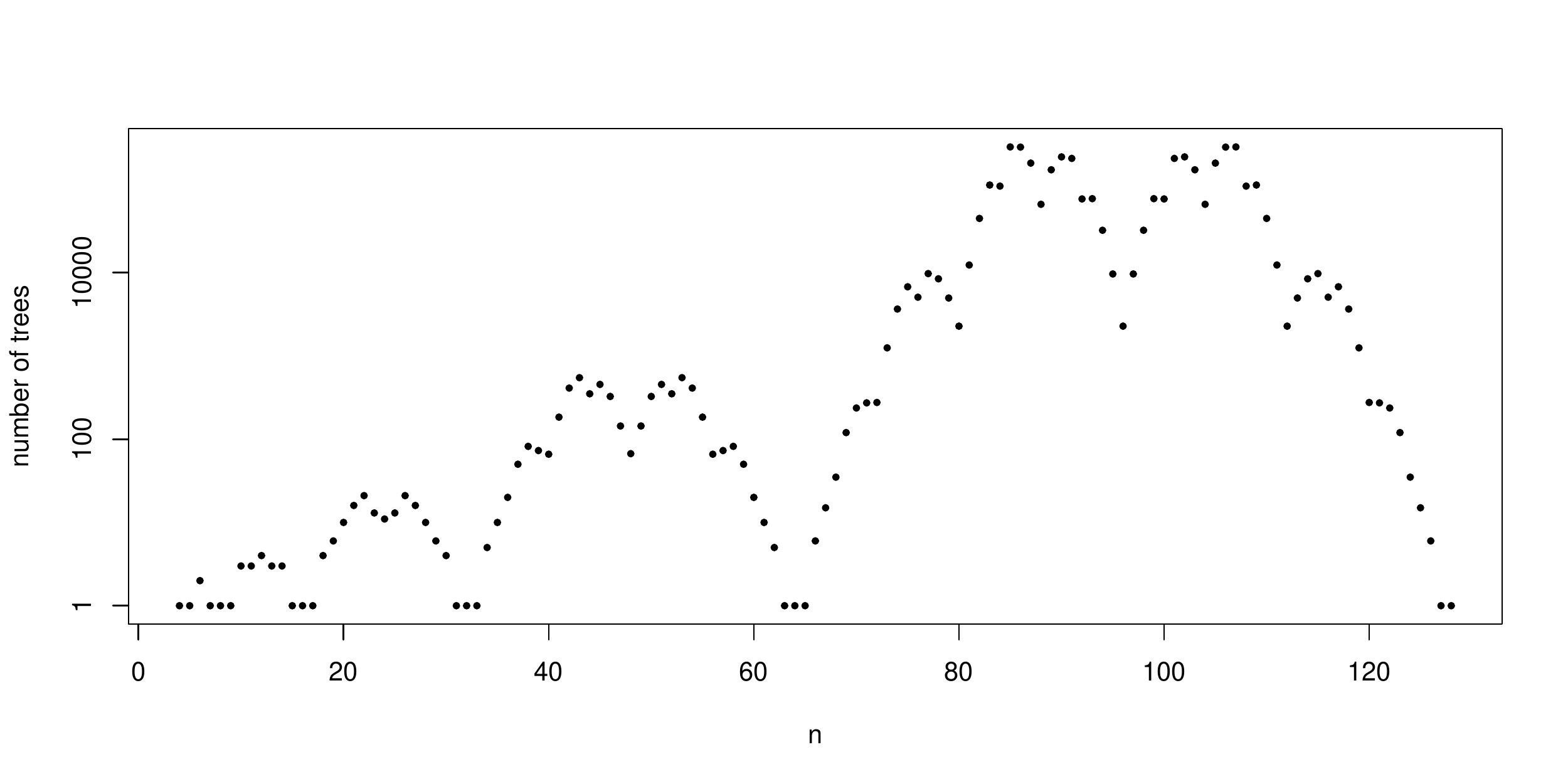}
\end{center}
\caption{\label{fig:1} 
Scatterplot of of the number of trees in $\TT_n$ with minimum Colless index, for $n=4,\ldots,128$.}
\end{figure}

\section{Conclusions}

The Colless index $C(T)$ of a bifurcating phylogenetic tree $T$ measures the total amount of imbalance of $T$, and it is one of the oldest and most popular balance indices for rooted bifurcating phylogenetic trees. But, despite its popularity, neither its minimum value for any given number of leaves nor the trees where this minimum value is reached were known so far. This paper fills this gap in the literature, with two main contributions. 

First, we have proved  that, for every number $n$ of leaves, the minimum Colless index of a bifurcating phylogenetic tree with $n$ leaves is reached at the maximally balanced bifurcating trees and we have provided an explicit formula for its vale $C(n)$. Knowing this minimum value, as well as its maximum value, which is reached at the combs and it is  $(n-1)(n-2)/2$, allows one to normalize the Colless index so that its range becomes the unit interval $[0,1]$, by means of the usual affine transformation
$$
\widetilde{C}(T)=\frac{C(T)-C(n)}{\frac{(n-1)(n-2)}{2}-C(n)}.
$$
This normalized index allows the comparison of the balance of trees with different numbers of leaves, which cannot be done directly with the unnormalized Colles index $C$ because its value tends to grow with $n$. 

The fact that the maximally balanced bifurcating trees have the minimum Colless index for their  number of leaves is not surprising, because, in words of \citet{Shao:90}, they are considered  the  ``most balanced'' bifurcating trees. But it turns out that for almost all values of  $n$ there are bifurcating trees with $n$ leaves that are not maximally balanced but whose Colless index is also minimal. So, our second main contribution is an alternative characterization of these trees and an efficient algorithm to produce all of them for any number $n$ of leaves. Notice that, in spite of not being considered the  ``most balanced'' ones because they have internal nodes whose imbalance is not minimal, the fact is that these trees also have  the minimum total amount of imbalance. So, the total imbalance of a phylogenetic tree does not capture the local imbalance at each internal node. 

The Colless index shares the drawback of classifying as ``most balanced'' many bifurcating trees that are not maximally balanced  with the other most popular balance index, the Sackin index \citep{Fischer19,Sackin:72,Shao:90}. This problem can be avoided by also using other balance indices,  like the total cophenetic index \citep{MRR1} or the rooted quartet index \citep{CMR}. Let us recall in this connection that  \citet{Shao:90} already advised to use more than one balance index to quantify tree balance.

We find remarkable the regularities displayed by the sequence of numbers of trees in $\TT_n$ with Colless index $C(n)$ hinted in Fig. \ref{fig:1} and that continue for larger values of $n$. We plan to study in  a future paper the properties of this sequence.\medskip

\noindent\textbf{Acknowledgements.}
 This research was partially supported by the Spanish Ministry of Economy and Competitiveness and the European Regional Development Fund through projects DPI2015-67082-P and PGC2018-096956-B-C43 (MINECO/FEDER).


%

\end{document}